\documentclass[a4paper,11pt]{article}
\usepackage[margin=1in]{geometry}
\usepackage{epic}
\usepackage{amsfonts,amssymb,amsmath}
\usepackage{enumitem}
\usepackage{tikz}
\usepackage{hyperref}

\newcommand{\tuple}[1]{\ensuremath{(#1)}}

\newcommand{\CG}[1]{\ensuremath{G_{#1}}} 
\newcommand{\AC}[1]{\ensuremath{\operatorname{AC}(#1)}} 
\newcommand{\PI}[1]{\ensuremath{\operatorname{Patt}(#1)}}  
\newcommand{\PG}[1]{\ensuremath{\operatorname{Patt}(#1)}}  
\newcommand{\SP}[2]{\ensuremath{#1 \stackrel{SP}{\rightarrow} {#2}}}  
\newcommand{\TM}[2]{\ensuremath{#1 \stackrel{TM}{\rightarrow} {#2}}}  

\newcommand{\ForbSP}[1]{\ensuremath{\operatorname{CSP}_{\rm\overline{SP}}(#1)}} 
\newcommand{\ForbTM}[1]{\ensuremath{\operatorname{CSP}_{\rm\overline{TM}}(#1)}} 

\DeclareMathOperator{\coNP}{coNP}
\DeclareMathOperator{\FPT}{FPT}
\DeclareMathOperator{\WW}{W[1]}
\DeclareMathOperator{\NP}{NP}

\usepackage{theorem}

\newtheorem{theorem}{Theorem}[section]
\newtheorem{definition}[theorem]{Definition}
\newtheorem{construction}[theorem]{Construction}
\newtheorem{proposition}[theorem]{Proposition}
\newtheorem{lemma}[theorem]{Lemma}
\newtheorem{corollary}[theorem]{Corollary}
\newtheorem{exmp}[theorem]{Example}

\newcommand{\bull}{\rule{.85ex}{1ex} \par \bigskip}
\newcommand{\node}{\makebox(0,0){$\bullet$}}
\newenvironment{proof}{\noindent {\bf Proof:\ }}{\hfill \bull}
\newenvironment{example}{\begin{exmp} \rm }{\hfill $\Box$ \end{exmp}}

\thicklines 
\newsavebox{\varfive}
\savebox{\varfive}(20,80){
\begin{picture}(20,80)(0,0)
\put(10,40){\oval(18,78)} \put(10,10){\makebox(0,0){$\bullet$}}
\put(10,25){\makebox(0,0){$\bullet$}}
\put(10,40){\makebox(0,0){$\bullet$}}
\put(10,55){\makebox(0,0){$\bullet$}}
\put(10,70){\makebox(0,0){$\bullet$}}
\end{picture}
}
\newsavebox{\varfour}
\savebox{\varfour}(20,50){
\begin{picture}(20,50)(0,0)
\put(10,25){\oval(18,48)} \put(10,10){\makebox(0,0){$\bullet$}}
\put(10,20){\makebox(0,0){$\bullet$}}
\put(10,30){\makebox(0,0){$\bullet$}}
\put(10,40){\makebox(0,0){$\bullet$}}
\end{picture}
}
\newsavebox{\varthree}
\savebox{\varthree}(20,40){
\begin{picture}(20,40)(0,0)
\put(10,20){\oval(18,38)} \put(10,10){\makebox(0,0){$\bullet$}}
\put(10,20){\makebox(0,0){$\bullet$}}
\put(10,30){\makebox(0,0){$\bullet$}}
\end{picture}
}
\newsavebox{\varii}
\savebox{\varii}(20,40){
\begin{picture}(20,40)(0,0)
\put(10,20){\oval(20,40)} \put(10,10){\makebox(0,0){$\bullet$}}
\put(10,30){\makebox(0,0){$\bullet$}}
\end{picture}
}
\newsavebox{\vartwo}
\savebox{\vartwo}(20,40){
\begin{picture}(20,40)(0,0)
\put(10,20){\oval(18,38)} \put(10,10){\makebox(0,0){$\bullet$}}
\put(10,30){\makebox(0,0){$\bullet$}}
\end{picture}
}
\newsavebox{\varone}
\savebox{\varone}(20,40){
\begin{picture}(20,40)(0,0)
\put(10,20){\oval(18,28)} \put(10,20){\makebox(0,0){$\bullet$}}
\end{picture}
} \thinlines

\begin{document}

\title{
Binary Constraint Satisfaction Problems \\
Defined by Excluded Topological Minors\thanks{An extended abstract of part of
this work appeared in the \emph{Proceedings of the 24th International Joint
Conference on Artificial Intelligence} (IJCAI'15)~\cite{ccjz15:ijcai}. The
authors were supported by EPSRC grant EP/L021226/1. Stanislav \v{Z}ivn\'y  was
supported by a Royal Society University Research Fellowship. This project has
received funding from the European Research Council (ERC) under the European
Union's Horizon 2020 research and innovation programme (grant agreement No
714532). The paper reflects only the authors' views and not the views of the ERC
or the European Commission. The European Union is not liable for any use that
may be made of the information contained therein.}}
 
\author{David A. Cohen \\ 
Royal Holloway, University of London\\  
\texttt{dave@cs.rhul.ac.uk}
\and
Martin C. Cooper \\ 
IRIT, University of Toulouse\\
 \texttt{cooper@irit.fr} \\
\and
Peter G. Jeavons \\
University of Oxford\\
\texttt{peter.jeavons@.cs.ox.ac.uk}
 \and
Stanislav \v{Z}ivn\'y \\
University of Oxford\\
\texttt{standa.zivny@cs.ox.ac.uk}
}

\date{}
\maketitle

\begin{abstract}
The  binary Constraint Satisfaction Problem (CSP) is to decide
whether there exists an assignment to a set of variables which
satisfies specified constraints between pairs of variables. 
A binary CSP instance can be presented as a labelled graph 
encoding both the forms of the constraints and where they are imposed. 
We consider subproblems defined by restricting the
allowed form of this graph. One type of restriction that has
previously been considered is to forbid certain specified
substructures (patterns). This captures some tractable classes of the
CSP, but does not capture classes defined by language restrictions,
or the well-known structural property of acyclicity. 

In this paper we extend the notion of pattern and introduce the notion of a
topological minor of a binary CSP instance. By forbidding a \emph{finite} set of
patterns from occurring as topological minors we obtain a compact mechanism for
expressing novel tractable subproblems of the binary CSP, including new
generalisations of the class of acyclic instances. Forbidding a finite set of
patterns as topological minors also captures all other tractable structural
restrictions of the binary CSP. Moreover, we show that several patterns give
rise to tractable subproblems if forbidden as topological minors but not if
forbidden as sub-patterns. Finally, we introduce the idea of augmented patterns
that allows for the identification of more tractable classes, including all
language restrictions of the binary CSP.
  \end{abstract}

\section{Introduction}  \label{sec:intro}

The Constraint Satisfaction Problem (CSP) is to decide whether it is
possible to find an assignment to a set of variables which satisfies constraints between certain subsets of the variables. This paradigm has been
applied in diverse application areas such as Artificial Intelligence,
Bioinformatics and Operations Research~\cite{Rossi06:handbook,Hell08:survey}.

As the CSP is known to be $\NP$-complete, much theoretical work has been
devoted to the identification of tractable subproblems. Important
tractable cases have been identified by restricting the hypergraph \textit{structure} of the constrained subsets
of variables~\cite{Freuder82:backtrack-free,Dalmau02:width}. 
Other tractable cases have been identified by restricting the forms
of constraints (sometimes called the constraint {\em
language})~\cite{DBLP:journals/jacm/JeavonsCG97,Feder98:monotone}.
Work on both of these areas is now essentially complete: full complexity
classifications have been established for all
structural restrictions~\cite{Grohe07:otherside,DBLP:journals/jacm/Marx13} and all language
restrictions~\cite{Bulatov17:focs,Zhuk17:focs}.

However, identifying the subproblems of the CSP that can be obtained by restricting either 
the structure or the language alone is not a sufficiently rich framework in which to investigate 
the full complexity landscape.  
For example, we may wish to identify all the instances solved by a particular 
algorithm, such as enforcing arc-consistency~\cite{Dechter03:processing,Rossi06:handbook}.
It has been shown~\cite{Feder98:monotone,Cohen16:GACdecides} 
that this class of instances includes all instances defined by 
a certain structural restriction, 
together with all instances defined by a certain language restriction,
as well as further instances that are not defined by either kind of restriction alone.
Hence we need a more flexible mechanism for describing subproblems that will 
allow us to unify and generalise such descriptions.

Here we develop a new mechanism of this kind that uses certain tools
from graph theory to define restricted classes of labelled graphs 
that represent binary CSP instances.
Our mechanism allows us to 
impose simultaneous restrictions on both the structure and the language of an instance,
and hence obtain a more refined collection of subproblems,
allowing a more detailed complexity analysis.
Subproblems of the CSP of this kind are sometimes referred to as \emph{hybrid} 
subproblems and, currently, very little is known about the complexity of
such subproblems~\cite{cz17:survey}.

The tools that we use to obtain restricted classes of labelled graphs build on a
well-established line of research in graph theory, by considering local
``obstructions" or ``forbidden patterns". The idea of using forbidden patterns
has previously been applied to the binary CSP and resulted in the discovery of a
number of new tractable
classes~\cite{Cohen12:pivot,ccez15:jcss,Cooper15:dam,EscamocherThesis}; related
ideas also appeared in~\cite{Madelain07:sicomp,Kun08:forbidden}.
In more detail, \cite{Cohen12:pivot} characterised all so-called negative
patterns that give rise to tractable classes of binary CSPs (this result is summarised in 
Theorem~\ref{thm:pivotSPtractable} below). Moreover, \cite{Cooper15:dam} characterised
all patterns consisting of at most two constraints that give rise to tractable
classes of binary CSPs. Finally, \cite{ccez15:jcss} investigated the notion of
forbidden patterns in the context of variable and value elimination in CSPs.

However, the existing theory of forbidden patterns is not sufficient to capture all known tractable 
structural restrictions, or language restrictions, as we show below.
In particular, we show that even the simplest tractable structural class,
the class of tree-structures CSP instances, cannot be captured by forbidding any finite 
set of patterns (Corollary~\ref{cor:acyclic}).
To describe all the relevant structural, language 
and hybrid restrictions that can ensure tractability therefore requires a more flexible way to define
restricted classes of instances.

In graph theory it proved useful to go beyond the idea of forbidden subgraphs 
and introduce the more flexible concept of forbidden minors. 
A well-known result of Robertson and Seymour states that any set of graphs closed 
under the operation of taking minors is specified by a finite set of forbidden minors.
Rather than adapting the full machinery of graph minors to the CSP framework, 
we consider here the slightly simpler notion of a 
\emph{topological minor}~\cite{Diestel10:graph}.
We show that by adapting the notion of topological minor to the CSP framework 
we are able to provide a unified description of all tractable structural classes,
all tractable language classes, and some new hybrid tractable classes that cannot be captured 
as either structural classes or language classes. 
Moreover, we are able to show that the class of tree-structured CSP instances
has a very simple description in this framework, and 
there exist tractable classes of the binary CSP
that properly extend this class and yet still have a very simple description.

The structure of the paper is as follows: 
in Section~\ref{sec:defs} we define the CSP and the notion of a pattern, and show how to associate
each CSP instance with a corresponding pattern.
In Section~\ref{sec:forbidding} we define what it means for a pattern to occur in another pattern, 
either as a {\em sub-pattern} or as a {\em topological minor}, and use these notions to define
restricted classes of CSP instances where specified patterns are forbidden from occurring in 
one or other of these ways. 

In Section~\ref{sec:structural} we show that all tractable structural classes of the CSP can be
characterised by forbidding certain patterns from occurring as topological minors.
We extend this idea in Section~\ref{sec:scheme} to obtain novel hybrid tractable classes of CSP instances,
including classes that properly extend the class of acyclic instances.

In Section~\ref{sec:detecting} we consider the complexity of determining whether a given pattern occurs as 
a topological minor in a CSP instance, and in Section~\ref{sec:augmented} we 
show that including additional structure in patterns allows us to characterise
more classes of CSP instances, including all tractable language classes.
Finally, in Section~\ref{sec:conclusion}, we conclude with a discussion of our results and present some 
open questions.

\section{Preliminaries} \label{sec:defs}

\subsection{The CSP}\label{sec:csp}

Constraint satisfaction is a paradigm for describing computational problems.
Each problem instance is represented as a constraint network: 
a collection of variables that take their value from some given domain.  
Some subsets of the variables have a further restriction on their allowed simultaneous
assignments, called a constraint.  A solution to such a network assigns a value 
to each variable such that every constraint is satisfied.  

In this paper we consider only \emph{binary} constraint networks,
where every constraint limits the possible assignments of precisely two variables.
It has been shown that 
any constraint network can be reduced to an equivalent binary network over a different
domain of values~\cite{Dechter89:tree,Rossi90:equivalence}.

\begin{definition}
\label{def:CSP}
An instance of the binary constraint satisfaction problem (CSP) 
is a triple \tuple{V,D,C} where $V$ is a finite set of variables, 
for each $v \in V$, $D(v)$ is a finite domain of values for $v$, 
and $C$ is a set of constraints, 
containing a constraint $R_{uv}$ for each pair of variables  $\tuple{u,v}$.  
The constraint $R_{uv} \subseteq D(u)\times D(v)$
is the set of compatible assignments to the variables $u$ and $v$.

A \emph{solution} to a binary CSP instance is an assignment $s:V\to D$ of values to variables
such that, for each constraint $R_{uv}$, $\tuple{s(u),s(v)} \in R_{uv}$.
\end{definition}

We will assume that there is \emph{exactly one} binary constraint between any two variables. 
That is, if we define $R'_{uv}$ as
$\{\tuple{b,a}\mid\tuple{a,b}\in R_{uv}\}$,
then $R_{vu} = R'_{uv}$.  
This is just a notational convenience since we can pre-process each instance, 
replacing $R_{uv}$ with $R_{uv} \cap R'_{vu}$. 
A constraint will be called \emph{trivial} if it is equal to 
the Cartesian product of the domains of its two variables.

The size of a CSP instance will be taken to be the sum of the sizes of the 
constraint relations. Given a fixed bound on the size of the domain for any variable 
and the arity of the constraints, this is polynomial in the number of variables. 
We will say that a class of CSP instances is \emph{tractable} if there is 
a polynomial-time algorithm to decide whether any instance in the class has a solution.

Note that Definition~\ref{def:CSP} describes a standard 
form of mathematical specification for a CSP instance 
that is convenient for theoretical analysis.
In the next subsection we will introduce an alternative representation 
in terms of patterns (see Construction~\ref{constructPI}).
Often more concise representations are used,
and trivial constraints are usually not represented~\cite{Rossi06:handbook}.

Arc-consistency (AC) is a fundamental concept for the binary
CSP~\cite{Dechter03:processing,Rossi06:handbook}.
\begin{definition}
A pair of variables $(u,v)$ is said to be arc-consistent if for each value
$a\in D(u)$ in the domain of $u$, there is a value $b\in D(v)$ in the domain of $v$ such that
$\tuple{a,b} \in R_{uv}$.  

A binary CSP instance is \emph{arc-consistent} if every pair of variables is arc-consistent.
\end{definition}
Given an arbitrary CSP instance $I$ there is a unique minimal set of domain values which can 
be removed to make the instance arc-consistent.  Furthermore the discovery of this unique
minimal set of domain values and their removal, called establishing arc-consistency, 
can be done in polynomial time~\cite{Cooper1989}.  For a given instance $I$ we will denote by
$\AC{I}$  the instance obtained after establishing arc-consistency.

\subsection{Patterns}
\label{sec:patterns}

We now introduce the central notion of a {\em pattern}, which can be thought of as a labelled 
graph, with three distinct kinds of edges.

\begin{definition}
\label{def:patt}
A \emph{pattern} is a structure $\tuple{X,E^{\sim},E^{+},E^{-}}$,
where 
{\samepage
\begin{itemize}
\item $X$ is a set of \emph{points};
\item $E^\sim$ is a binary equivalence relation over $X$ 
whose equivalence classes are called \emph{parts};
\item $E^{+}$ is a symmetric binary relation over $X$ whose tuples are called \emph{positive edges};
\item $E^{-}$ is a symmetric binary relation over $X$ whose tuples are called \emph{negative edges}.
\end{itemize}
}
The sets $E^\sim$ and $E^{+}$ are disjoint, and the sets $E^\sim$ and $E^{-}$ are disjoint.
\end{definition}

In a general pattern 
there may be pairs of points $x$ and $y$ in distinct parts such that 
$(x,y)$ is neither a positive nor a negative edge, 
and there may be pairs of points $x$ and $y$ in distinct parts such that 
$(x,y)$ is \emph{both} a positive and a negative edge.
A pattern is called \emph{complete}
if every pair of points $x$ and $y$ in distinct parts are connected by either a positive
or negative edge (but not both), 
and hence $E^{\sim} \cup E^{+} \cup E^{-} = X^2$.

\begin{example}
Some examples of patterns are illustrated in a standard way in Figure~\ref{fig:sp}.

The pattern shown in Figure~\ref{fig:sp}(a) is complete, but the others are not.
\thicklines \setlength{\unitlength}{0.7pt}
\begin{figure}
\centering

\begin{picture}(480,110)(0,0)

\put(0,10){
\begin{picture}(130,100)(0,0)
\put(0,0){\usebox{\varone}} \put(80,0){\usebox{\varone}}
\put(40,40){\usebox{\varone}} \put(10,20){\line(1,0){80}}
\dashline[50]{7}(10,20)(50,60) \dashline[50]{7}(50,60)(90,20)
\put(50,0){\makebox(0,0){(a)}}
\end{picture}}

\put(120,10){\begin{picture}(130,100)(0,0)
\put(0,0){\usebox{\varone}} \put(80,0){\usebox{\varone}}
\put(40,50){\usebox{\varthree}} \put(10,20){\line(1,1){40}}
\put(10,20){\line(1,0){80}} \put(90,20){\line(-1,1){40}}
\dashline[50]{7}(10,20)(50,70) \dashline[50]{7}(50,80)(90,20)
\put(50,0){\makebox(0,0){(b)}}
\end{picture}}

\put(240,10){\begin{picture}(100,60)(0,0) 
\put(0,0){\usebox{\varone}}
\put(80,0){\usebox{\varone}} \dashline[50]{7}(10,20)(90,20)
\qbezier(10,20)(50,30)(90,20) \put(50,0){\makebox(0,0){(c)}}
\end{picture}}

\put(360,10)
{\begin{picture}(130,100)(0,0)
\put(0,0){\usebox{\varone}} \put(80,0){\usebox{\varone}}
\put(40,50){\usebox{\varthree}} \put(10,20){\line(1,1){40}}
\put(90,20){\line(-1,1){40}} \dashline[50]{7}(10,20)(50,70)
\dashline[50]{7}(50,80)(90,20) \put(50,0){\makebox(0,0){(d)}}
\end{picture}}

\end{picture}

\caption{Some example patterns.
Points are shown as filled circles, parts as ovals, 
positive edges as solid lines and negative edges as dashed lines.
}
\label{fig:sp}
\end{figure}
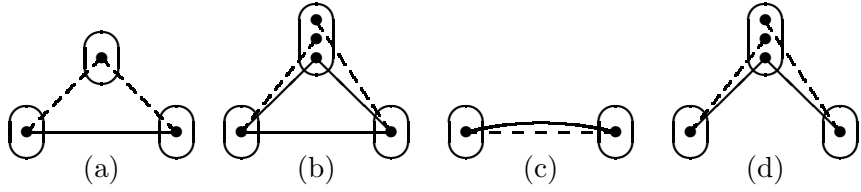
\end{example}

It will often be convenient to build special patterns to represent binary CSP instances, 
so we now define the following construction.
\begin{construction}
\label{constructPI}
For any binary CSP instance $I=\tuple{V,D,C}$, 
where $C = \{R_{uv} \mid u,v \in V, u \neq v\}$,
we define a corresponding complete pattern 
$\PI{I} = \tuple{X,E^{\sim},E^{+},E^{-}}$ where
\begin{itemize}
\item 
$X = \{x_{v,a} \mid v \in V, a \in D(v)\}$;
\item
$E^\sim = \{(x_{u,a},x_{v,b})\in X\times X \mid u = v\}$;
\item
$E^{+} = \{(x_{u,a},x_{v,b})\in X\times X \mid u \neq v, \tuple{a,b} \in R_{uv}\}$;
\item
$E^{-} = \{(x_{u,a},x_{v,b})\in X\times X \mid u \neq v, \tuple{a,b} \not\in R_{uv}\}$.
\end{itemize}
\end{construction}
We remark that for any instance $I$ the points of $\PI{I}$ are the possible 
assignments for each individual variable, and the parts of $\PI{I}$ 
correspond to sets of possible assignments for a particular variable.
Positive edges in $\PI{I}$ 
correspond to allowed pairs of assignments and
are therefore closely related to the edges of the 
\emph{microstructure} representation of $I$ defined in~\cite{Jegou93:microstructure};
negative edges correspond to disallowed pairs of assignments and
are closely related to the edges of the 
\emph{microstructure complement} discussed in~\cite{Cohen2003d}.
\begin{example}
Figure~\ref{fig:sp}(a) shows the pattern $\PI{I}$
for a rather trivial instance $I$ with three variables,
each of which has only one possible value. 
Note that $I$ has no solution because the only possible assignments for two pairs of variables 
are in negative edges and hence disallowed by the constraints.
\end{example}

A pattern with no positive edges will be called a \emph{negative pattern}.
It will sometimes be convenient to build negative patterns from graphs,
so we now define the following construction.
\begin{construction}%
\label{constructPG}
For any graph $G=\tuple{V,E}$, we define a corresponding negative pattern 
$\PG{G} = \tuple{X,E^{\sim},\emptyset,E^{-}}$ where
\begin{itemize}
\item 
$X = \{x_{e,v} \mid e \in E, v \in e\}$;
\item
$E^\sim = \{(x_{e,u},x_{f,v})\in X\times X \mid u = v\}$;
\item
$E^{-} = \{(x_{e,u},x_{f,v})\in X\times X \mid e = f, u \neq v\}$.
\end{itemize}
\end{construction}
\begin{example}
Let $C_3$ be the 3-cycle, that is, the graph with three vertices, $v_1,v_2,v_3$,
and 3 edges $e_1,e_2,e_3$, where $e_1 = \{v_1,v_2\}, e_2 = \{v_2,v_3\}$ and $e_3 = \{v_3,v_1\}$.
The associated negative pattern $\PG{C_3}$ defined by Construction~\ref{constructPG}
is the pattern with 6 points, 3 parts, and 3 negative edges,
shown in  Figure~\ref{fig:acyclic}.
\thicklines \setlength{\unitlength}{0.7pt}
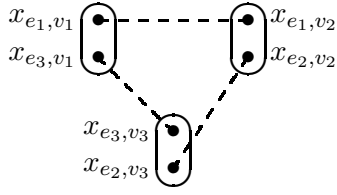
\begin{figure}
\centering

\begin{picture}(160,110)(-30,-10)
\put(0,60){\usebox{\vartwo}}
\put(40,0){\usebox{\vartwo}}
\put(80,60){\usebox{\vartwo}}
\dashline[50]{6}(10,70)(50,30)
\dashline[50]{6}(10,90)(90,90)
\dashline[50]{6}(50,10)(90,70)
\put(-20,90){\makebox(0,0){$x_{e_1,v_1}$}}
\put(-20,70){\makebox(0,0){$x_{e_3,v_1}$}}
\put(120,90){\makebox(0,0){$x_{e_1,v_2}$}}
\put(120,70){\makebox(0,0){$x_{e_2,v_2}$}}
\put(20,30){\makebox(0,0){$x_{e_3,v_3}$}}
\put(20,10){\makebox(0,0){$x_{e_2,v_3}$}}
\end{picture}

\caption{The pattern $\PG{C_3}$ constructed from the cycle graph $C_3$
by Construction~\ref{constructPG}.
}
\label{fig:acyclic}
\end{figure}

\thicklines \setlength{\unitlength}{0.7pt}
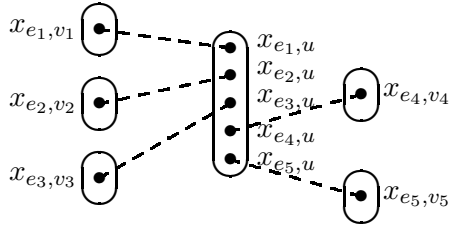
\begin{figure}
\centering

\begin{picture}(220,130)(0,-20)
\put(30,-10){\usebox{\varone}}
\put(30,30){\usebox{\varone}}
\put(30,70){\usebox{\varone}}
\put(170,-20){\usebox{\varone}}
\put(170,35){\usebox{\varone}}
\put(100,10){\usebox{\varfive}}
\dashline[50]{6}(40,10)(110,50)
\dashline[50]{6}(40,50)(110,65)
\dashline[50]{6}(40,90)(110,80)
\dashline[50]{6}(180,55)(110,35)
\dashline[50]{6}(180,0)(110,20)
\put(10,10){\makebox(0,0){$x_{e_3,v_3}$}}
\put(10,50){\makebox(0,0){$x_{e_2,v_2}$}}
\put(10,90){\makebox(0,0){$x_{e_1,v_1}$}}
\put(210,0){\makebox(0,0){$x_{e_5,v_5}$}}
\put(210,55){\makebox(0,0){$x_{e_4,v_4}$}}
\put(140,82){\makebox(0,0){$x_{e_1,u}$}}
\put(140,67){\makebox(0,0){$x_{e_2,u}$}}
\put(140,52){\makebox(0,0){$x_{e_3,u}$}}
\put(140,33){\makebox(0,0){$x_{e_4,u}$}}
\put(140,18){\makebox(0,0){$x_{e_5,u}$}}
\end{picture}

\caption{The pattern $\PG{K_{1,5}}$ constructed from the star graph $K_{1,5}$
by Construction~\ref{constructPG}.
}
\label{fig:K15}
\end{figure}

Let $K_{1,k}$ be a star graph with $k$ leaves; that is, the graph with vertices $\{u,v_1,\ldots,v_k\}$ 
and edges $\{u,v_i\}$ for $1\leq i\leq k$. The pattern $\PG{K_{1,k}}$ has $2k$ points, $k+1$ parts, and $k$ negative edges. The case of $k=5$ is shown in Figure~\ref{fig:K15}.
\end{example}

In graph theory, a \emph{subdivision} operation on a graph replaces an edge $(u,v)$ with 
a path of length two by introducing a new vertex $z_{uv}$, and connecting $u$ to $z_{uv}$ 
and $z_{uv}$ to $v$~\cite{Diestel10:graph}.
A graph $G$ is said to be a topological minor of a graph $H$ if some sequence of subdivision
operations on $G$ yields a subgraph of $H$~\cite{Diestel10:graph}.
We now define an operation on patterns that is analogous to the subdivision operation on graphs,
but takes into account the three different types of edges that are present in a pattern.
This subdivision operation for patterns 
is crucial to the idea of defining topological minors in patterns,
as described in Section~\ref{sec:forbidding}.

\begin{definition}
\label{def:subdivision}
Let $P = \tuple{X,E^{\sim},E^{+},E^{-}}$ be a pattern.

For any two distinct parts $U,V$ of $P$, 
we define $E^+_{UV} = E^+ \cap (U \times V)$,
$E^-_{UV} = E^- \cap (U \times V)$,
and $Z_{UV} = \{z_{xy} \mid (x,y) \in E^+_{UV} \} \cup \{z'_{xy},z''_{xy} \mid (x,y) \in E^-_{UV}\}$.
The \emph{subdivision} of $P$ at $U,V$ 
is defined to be the pattern $P_d = \tuple{X_d,E_d^{\sim},E_d^{+},E_d^{-}}$ where
\begin{itemize}
\item 
$X_d = X \cup Z_{UV}$;
\item 
$E_d^\sim = E^\sim \cup (Z_{UV} \times Z_{UV})$;
\item
$
\begin{aligned}[t]
E_d^+ = (E^+ \setminus & \{(x,y),(y,x) \mid (x,y) \in E^+_{UV} \})\\
 & \cup \{(x,z_{xy}),(z_{xy},x),(z_{xy},y),(y,z_{xy}) \mid (x,y) \in E^+_{UV}\};
\end{aligned}
$
\item
$
\begin{aligned}[t]
E_d^- = (E^- \setminus & \{(x,y),(y,x) \mid (x,y) \in E^-_{UV}\})\\
 & \cup \{(x,z'_{xy}),(z'_{xy},x),(z''_{xy},y),(y,z''_{xy}) \mid (x,y) \in E^-_{UV}\}.
\end{aligned}
$
\end{itemize}
Pattern $P'$ is called a \emph{subdivision} of $P$ if it can be obtained from $P$ 
by some (possibly empty) sequence of subdivision operations.
\end{definition}
\begin{example}
\label{ex:subdivision}
The pattern shown in Figure~\ref{fig:sp}(d) can be obtained by performing 
a single subdivision operation on the pattern shown in Figure~\ref{fig:sp}(c).
\end{example}

We remark that positive and negative edges are treated differently in
Definition~\ref{def:subdivision}: a single extra point, $z_{xy}$, is added for each 
positive edge $(x,y)$, and \emph{two} extra points, $z'_{xy}$ and $z''_{xy}$, 
are added for each negative edge (see Example~\ref{ex:subdivision}).
This difference reflects a semantic difference between positive and negative edges 
in a CSP instance, which we illustrate as follows.
Suppose that the assignment of value $a$ to variable $u$ and value $b$ to variable $v$ extends to a solution. 
In this case, for any other variable $w$,
the points $\tuple{u,a}$ and $\tuple{v,b}$ must both be compatible with some common point
$\tuple{w,c}$.  On the other hand, the assignment of $a$ to variable $u$ and $b$ to variable $v$ 
may fail to extend to a solution if there are points \tuple{w,c} and \tuple{w,d} where 
\tuple{u,a} is incompatible with \tuple{w,c}, \tuple{v,b} is incompatible with \tuple{w,d} 
and the rest of the instance forces $w$ to take either value $c$ or value $d$.

\section{Forbidding patterns}
\label{sec:forbidding}

In the remainder of this paper we consider classes of binary CSP instances that are 
defined by \emph{forbidding} a specified set of patterns from occurring in certain ways,
which we now define.

\subsection{Occurrences of one pattern in another} 

\begin{definition}
\label{def:patternhomomorphism}
A pattern $P_1 = \tuple{X_1,E_1^{\sim},E_1^{+},E_1^{-}}$ is said to 
have a \emph{homomorphism} to 
a pattern $P_2 = \tuple{X_2,E_2^{\sim},E_2^{+},E_2^{-}}$,
if there is a mapping $h:X_1 \rightarrow X_2$ such that 
\begin{itemize}
\item 
if $(x,y) \in E_1^\sim$ then $(h(x),h(y)) \in E_2^\sim$, and
\item
if $(x,y) \in E_1^{+}$ then $(h(x),h(y)) \in E_2^{+}$, and
\item
if $(x,y) \in E_1^{-}$ then $(h(x),h(y)) \in E_2^{-}$.
\end{itemize}
\end{definition}
A homomorphism $h$ from a pattern $P_1 = \tuple{X_1,E_1^{\sim},E_1^{+},E_1^{-}}$
to a pattern $P_2 = \tuple{X_2,E_2^{\sim},E_2^{+},E_2^{-}}$ will be said to
\emph{preserve parts} if it satisfies the additional
property that for all $(x,y) \in X_1^2$, if  
$(x,y) \not\in E_1^\sim$, then $(h(x),h(y)) \not\in E_2^\sim$.
\begin{definition}
\label{def:sub-pattern}
A pattern $P_1$ is said to \emph{occur as a sub-pattern} in  
a pattern $P_2$, 
denoted $\SP{P_1}{P_2}$,
if there is a homomorphism from $P_1$ to $P_2$ that preserves parts.
\end{definition}

Earlier papers~\cite{Cohen12:pivot,Cooper15:dam} 
have defined the notions of pattern and the notion of occurring as a sub-pattern 
in slightly different ways, but these are all essentially equivalent to 
Definition~\ref{def:sub-pattern}.

\begin{example}
The pattern shown in Figure~\ref{fig:sp}(d) 
has a homomorphism to the pattern shown in Figure~\ref{fig:sp}(c),
but does not occur as a sub-pattern in this pattern.
The pattern shown in Figure~\ref{fig:sp}(d) 
does occur as a sub-pattern in the pattern shown in Figure~\ref{fig:sp}(b).
\end{example}

Now we introduce a new form of occurrence that will be our focus in this paper, 
and will allow us to define a wider range of restricted subproblems of the CSP.
\begin{definition}
\label{def:top-minor}
A pattern $P_1$ is said to 
\emph{occur as a topological minor}
in a pattern $P_2$,
denoted $P_1 \stackrel{TM}{\rightarrow} P_2$,
if some subdivision of $P_1$ occurs as a sub-pattern in $P_2$.
\end{definition}

\begin{example}
The pattern shown in Figure~\ref{fig:sp}(c) occurs as a topological minor in
the pattern shown in Figure~\ref{fig:sp}(d) and
in the pattern shown in Figure~\ref{fig:sp}(b).
\end{example}

\begin{lemma}
\label{lem:properties}
For any patterns $P, P'$ and $P''$ the following properties hold:
\begin{enumerate}[label=(\alph*)]
\item \label{lem:propreflexive} 
$\SP{P}{P}$ and $\TM{P}{P}$;
\item \label{lem:propSPimpliesTM} 
If $\SP{P}{P'}$, then $\TM{P}{P'}$;
\item \label{lem:propSPtrans} 
If $\SP{P}{P'}$ and $\SP{P'}{P''}$, then $\SP{P}{P''}$;
\item \label{lem:propTMtrans} 
If $\TM{P}{P'}$ and $\TM{P'}{P''}$, then $\TM{P}{P''}$.
\end{enumerate}
\end{lemma}
\begin{proof}
Part~\ref{lem:propreflexive} 
is obtained by taking the identity function as a homomorphism,
and an empty sequence of subdivisions.
Part~\ref{lem:propSPimpliesTM} 
is obtained by taking an empty sequence of subdivisions.
Part~\ref{lem:propSPtrans} 
is obtained by composing the two homomorphisms.

Part~\ref{lem:propTMtrans} 
follows from the following observation: 
assume that $h$ is a homomorphism from $P_1$ to $P_2$ that preserves parts, 
and that $P_3$ is the pattern obtained by performing a subdivision operation on $P_2$ 
at parts $U$ and $V$. 
Now consider the pattern $Q$ obtained by performing a 
subdivision operation on $P_1$ at the parts that are mapped by $h$ to $U$ and $V$.
By our definition of subdivision, it follows that $h$ can be extended 
to a homomorphism $h'$ from $Q$ to $P_3$ that preserves parts.

Hence in any sequence of subdivision operations and homomorphisms that preserve parts
we can re-order the operations to perform all subdivisions at the start, and then compose
all the homomorphisms.  
\end{proof}

Recall that establishing arc-consistency in an instance $I$ involves removing domain values from $I$ and yields the (unique) instance AC($I$), hence it 
cannot introduce an occurrence of a pattern as a sub-pattern or as a topological minor 
if it did not already occur. This gives the following result.
\begin{lemma}
\label{lem:AC}
For any patterns $P$ and $I$, where $I$ represents an instance, the following properties hold:
\begin{enumerate}
\item[(a)] If $\SP{P}{\PI{\AC{I}}}$, then $\SP{P}{\PI{I}}$;
\item[(b)] If $\TM{P}{\PI{\AC{I}}}$, then $\TM{P}{\PI{I}}$.
\end{enumerate}
\end{lemma}
Establishing arc-consistency can be done in polynomial time, 
so for many of our results we will only need to consider arc-consistent CSP instances.

\subsection{Restricted classes of instances}

We can use Definition~\ref{def:sub-pattern} to define restricted classes of 
binary CSP instances by forbidding the occurrence of certain patterns 
as sub-patterns in those instances.
\begin{definition}
\label{def:CSPSP}
Let $\cal S$ be a set of patterns.

We denote by \ForbSP{\cal S} the set
of all binary CSP instances $I$ such that for all $P \in {\cal S}$ 
it is not the case that $\SP{P}{\PI{I}}$.
\end{definition}

\begin{definition}
We will say that a pattern $P$ is \emph{sub-pattern tractable} if 
\ForbSP{\{P\}} is tractable;
we will say that a pattern $P$ is \emph{sub-pattern $\NP$-complete} 
if \ForbSP{\{P\}} is $\NP$-complete.
\end{definition}

For simplicity, we write \ForbSP{P} for \ForbSP{\{P\}}.

The complexity of the class \ForbSP{\cal S} has been determined
for a wide range of patterns~\cite{Cooper10:BTP,Cohen12:pivot,Cooper15:dam}.
In fact, for all \emph{negative} patterns $P$ 
the complexity of \ForbSP{P}
has been completely characterised~\cite{Cohen12:pivot}.
To define this characterisation, we need to introduce the idea of \emph{star patterns}.

A connected graph $G$ is called a \emph{star} if it is acyclic, and has 
exactly one vertex of degree greater than 2. 
The vertex of degree greater than 2 in a star graph will
be called the central vertex. 
A pattern $P$ will be called a \emph{star pattern} 
if it can be obtained from the pattern $\PG{G}$ for some star graph $G$ by merging
zero or more points in the part of $\PG{G}$ corresponding to the central vertex of $G$.

\begin{example}
\label{ex:starpatterns}
Since the empty graph is a star graph, 
the simplest star pattern is the empty pattern, which has no points.
Some other examples of star patterns are shown in Figure~\ref{fig:NPCstars}.
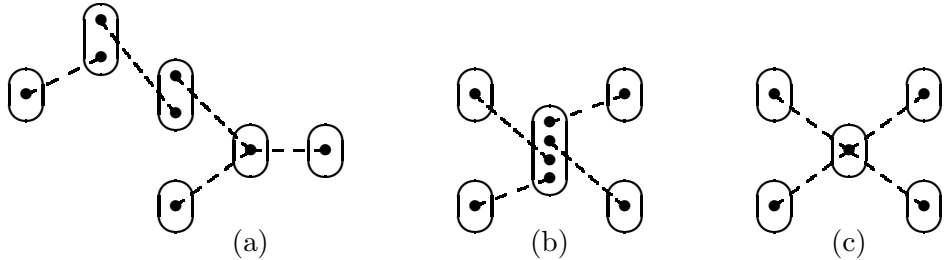
\begin{figure}[ht]
\thicklines \setlength{\unitlength}{0.7pt} 
\centering
\begin{picture}(500,140)(0,0)

\put(0,10){
\begin{picture}(180,130)(-80,0)
\put(0,0){\usebox{\varone}}
\put(0,60){\usebox{\vartwo}}
\put(-40,90){\usebox{\vartwo}}
\put(-80,60){\usebox{\varone}}
\dashline[50]{7}(-30,120)(10,70)
\dashline[50]{7}(-70,80)(-30,100)
\put(40,30){\usebox{\varone}} \put(80,30){\usebox{\varone}}
\dashline[50]{7}(10,20)(50,50)  \dashline[50]{7}(10,90)(50,50)
\dashline[50]{7}(50,50)(90,50) \put(50,0){\makebox(0,0){(a)}} \end{picture}}

\put(240,10){
\begin{picture}(100,90)(0,0)
\put(0,0){\usebox{\varone}} \put(0,60){\usebox{\varone}} \put(80,0){\usebox{\varone}} \put(80,60){\usebox{\varone}} \put(40,25){\usebox{\varfour}} \dashline[50]{7}(10,20)(50,35)
\dashline[50]{7}(10,80)(50,45) \dashline[50]{7}(50,55)(90,20)
\dashline[50]{7}(50,65)(90,80) \put(50,0){\makebox(0,0){(b)}} \end{picture}}

\put(400,10){
\begin{picture}(100,90)(0,0)
\put(0,0){\usebox{\varone}} \put(0,60){\usebox{\varone}} \put(80,0){\usebox{\varone}} \put(80,60){\usebox{\varone}} \put(40,30){\usebox{\varone}} \dashline[50]{7}(10,20)(50,50)
\dashline[50]{7}(10,80)(50,50) \dashline[50]{7}(50,50)(90,20)
\dashline[50]{7}(50,50)(90,80) \put(50,0){\makebox(0,0){(c)}} \end{picture}}

\end{picture}
\caption{Examples of star patterns.}
\label{fig:NPCstars}
\end{figure}
\end{example}

\begin{definition}[\cite{Cohen12:pivot}]
\label{def:pivotk}
For any $k \geq 1$, the star pattern with 3 branches, each of length $k$,
where exactly two points are merged in the central part, 
as shown in Figure~\ref{fig:pivot}, is called Pivot($k$).
\end{definition}

\begin{figure}[ht]
\bigskip
\bigskip
\thicklines \setlength{\unitlength}{0.7pt}
\centering
\begin{picture}(380,125) 
\put(0,80){\usebox{\varii}}
\put(80,80){\usebox{\varii}} \put(130,80){\usebox{\varii}}
\put(180,60){\usebox{\varii}} \put(230,80){\usebox{\varii}}
\put(280,80){\usebox{\varii}} \put(360,80){\usebox{\varii}}
\put(0,20){\usebox{\varii}} \put(80,20){\usebox{\varii}}
\put(130,20){\usebox{\varii}} 
\dashline[50]{7}(10,110)(35,100)
\dashline[50]{7}(65,100)(90,90) \dashline[50]{7}(90,110)(140,90)
\dashline[50]{7}(140,110)(190,90) \dashline[50]{7}(190,90)(240,110)
\dashline[50]{7}(240,90)(290,110) \dashline[50]{7}(290,90)(315,100)
\dashline[50]{7}(345,370,110) \dashline[50]{7}(10,30)(35,40)
\dashline[50]{7}(65,40)(90,50) \dashline[50]{7}(90,30)(140,50)
\dashline[50]{7}(140,30)(190,70) \dashline[50]{7}(345,100)(370,110)
\put(50,100){\makebox(0,0){.\;.\;.}} \put(50,40){\makebox(0,0){.\;.\;.}}
\put(330,100){\makebox(0,0){.\;.\;.}}
\put(0,0){\makebox(200,280){$\overbrace{\makebox(180,0){}}^{\mbox{$k$ edges}}$}}
\put(0,0){\makebox(200,0){$\underbrace{\makebox(180,0){}}_{\mbox{$k$ edges}}$}}
\put(0,0){\makebox(565,80){$\underbrace{\makebox(180,0){}}_{\mbox{$k$ edges}}$}}
\end{picture}
\vspace{0.3cm}
\caption{The pattern Pivot($k$). 
} 
\label{fig:pivot}
\setlength{\unitlength}{1pt}
\end{figure}
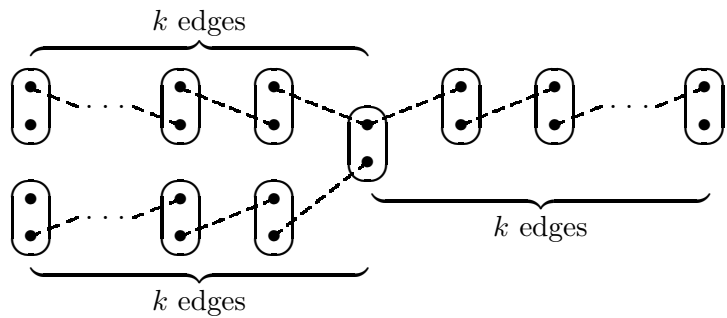

\begin{theorem}[\cite{Cohen12:pivot}]
\label{thm:pivotSPtractable}
For any $k \geq 1$, the negative pattern Pivot(k) shown in Figure~\ref{fig:pivot}
is sub-pattern tractable, as are all negative patterns $P$ such that
$\SP{P}{Pivot(k)}$; all other negative patterns are sub-pattern $\NP$-complete.
\end{theorem}
\begin{example}
By Theorem~\ref{thm:pivotSPtractable} all the negative patterns shown in Figures~\ref{fig:acyclic}, \ref{fig:K15} and \ref{fig:NPCstars} are sub-pattern NP-complete.
\end{example}
To go beyond the earlier results for forbidden 
sub-patterns~\cite{Cohen12:pivot,ccez15:jcss,Cooper15:dam,EscamocherThesis},
and define a wider range of restricted classes,
we use Definition~\ref{def:top-minor} 
to define restricted classes of binary CSP instances by forbidding 
the occurrence of certain patterns as topological minors in those instances.
\begin{definition}
\label{def:CSPTM}
Let $\cal S$ be a set of patterns.

We denote by \ForbTM{\cal S} the set
of all binary CSP instances $I$ such that for all $P \in {\cal S}$ 
it is not the case that $\TM{P}{\PI{I}}$.
\end{definition}

\begin{definition}
We will say that a pattern $P$ is \emph{topological-minor tractable} if 
\ForbTM{\{P\}} is tractable;
we will say that a pattern $P$ is \emph{topological-minor $\NP$-complete} if 
\ForbTM{\{P\}} is $\NP$-complete.
\end{definition}

For simplicity, we write \ForbTM{P} for \ForbTM{\{P\}}.

By Lemma~\ref{lem:properties}\,\ref{lem:propSPimpliesTM}, if $P$ occurs as a
sub-pattern of some pattern $Q$, then it also occurs as a topological minor of
$Q$.  Hence for any pattern $P$ we have that $\ForbTM{P} \subseteq \ForbSP{P}$.
The following is an immediate consequence.
\begin{lemma} \label{lem:sp} 
If a pattern $P$ is sub-pattern tractable then $P$
is also topological-minor tractable. 
\end{lemma}
\begin{example} \label{ex:1a1bSPandTMtractable}
By the results of earlier work, 
the two patterns shown in Figure~\ref{fig:sp}(a) and~\ref{fig:sp}(b) are known to be sub-pattern tractable: 
the tractability of the pattern shown in Figure~\ref{fig:sp}(a) follows from the tractability of a more
general pattern (called JWP) defined in~\cite{cz11:ai}; the tractability of the pattern 
shown in Figure~\ref{fig:sp}(b) follows from~\cite[Lemma~46]{EscamocherThesis} (where it corresponds to pattern $U'_{30}$).

Hence both patterns are also topological-minor tractable, by Lemma~\ref{lem:sp}.
\end{example}

By Lemma~\ref{lem:properties}\,\ref{lem:propTMtrans}, 
if $P$ occurs as a topological minor in $Q$ then 
$\ForbTM{P} \subseteq \ForbTM{Q}$. 
The following is an immediate consequence.
\begin{lemma} 
\label{lem:trans}
If pattern $P \stackrel{TM}{\rightarrow} Q$, and $Q$
is topological-minor tractable, then $P$ is also topological-minor tractable.
\end{lemma}
\begin{example}
\label{ex:figSPisTMtractable}
We can deduce from Lemma~\ref{lem:trans} that
Figure~\ref{fig:sp}(d) is topological-minor tractable, since
Figure~\ref{fig:sp}(d) occurs as a sub-pattern (and hence also as a topological minor) 
in Figure~\ref{fig:sp}(b),
and it was shown in Example~\ref{ex:1a1bSPandTMtractable} 
that Figure~\ref{fig:sp}(b) is topological-minor tractable.
\end{example}

The converse of Lemma~\ref{lem:sp} does not hold: there exist patterns that are topological-minor
tractable but sub-pattern NP-complete, as the following example demonstrates.
More significant examples will be discussed in Section~\ref{sec:scheme}.
\begin{example}
Figure~\ref{fig:sp}(c) is sub-pattern $\NP$-complete, since it 
cannot occur as a sub-pattern of any instance, so for this pattern P,
\ForbSP{P} contains all possible CSP instances.
However, by Lemma~\ref{lem:trans}, Figure~\ref{fig:sp}(c) is topological-minor tractable, 
since it occurs as a topological minor in Figure~\ref{fig:sp}(d),
and it was shown in Example~\ref{ex:figSPisTMtractable} 
that Figure~\ref{fig:sp}(d) is topological-minor tractable.
\end{example}

For some patterns $P$, the sets \ForbSP{P} and \ForbTM{P} are identical, 
as our next result shows.
A pattern $P$ will be called 
\emph{star-like} if removing the positive edges from $P$
gives a negative pattern $P'$ such that $\SP{P'}{P''}$ for some
star pattern $P''$.

\begin{example}
\label{ex:starlike}
All of the patterns shown in Figures~\ref{fig:sp}, \ref{fig:K15} and~\ref{fig:NPCstars} are star-like, 
but the pattern shown in Figure~\ref{fig:acyclic} is not star-like.
\end{example}

\begin{proposition}
\label{prop:star}
If $P$ is a star-like negative pattern, then $\ForbTM{P} = \ForbSP{P}$.
\end{proposition}
\begin{proof}
By Lemma~\ref{lem:properties}\,\ref{lem:propSPimpliesTM}, 
for any pattern $P$ we have that 
$\ForbTM{P} \subseteq \ForbSP{P}$. 

To obtain the reverse inclusion,
let $P$ be a star-like negative pattern, 
and let $Q$ be a star pattern such that $\SP{P}{Q}$.
By the definition of star pattern, for any subdivision $Q'$ of $Q$, 
we have that $\SP{Q}{Q'}$.
Hence, by Lemma~\ref{lem:properties}\,\ref{lem:propSPtrans}
$\SP{P}{Q'}$, so $\ForbSP{P} \subseteq \ForbSP{Q'}$.
But this implies, by Definition~\ref{def:top-minor}, that
$\ForbSP{P} \subseteq \ForbTM{P}$.
\end{proof}
\begin{example}
By Theorem~\ref{thm:pivotSPtractable}, any pattern Pivot($k$) is sub-pattern tractable,
and by Proposition~\ref{prop:star} we know that forbidding Pivot($k$) as a
topological minor defines the same set of instances as forbidding
Pivot($k$) as a sub-pattern. 
Therefore, for any $k \geq 1$, the pattern Pivot($k$) is also topological-minor tractable. 

Similarly,
by Theorem~\ref{thm:pivotSPtractable}, each star pattern $P$ 
shown in Figure~\ref{fig:NPCstars} is sub-pattern $\NP$-complete. 
By Proposition~\ref{prop:star}, for each of these patterns 
$\ForbTM{P} = \ForbSP{P}$.
Consequently, these patterns are also topological-minor $\NP$-complete.
\end{example}

We now give a partial converse of Proposition~\ref{prop:star}, 
by showing that for all patterns $P$ that are not star-like, 
\ForbTM{P} cannot be expressed by forbidding
any finite set of sub-patterns. This means that the notion of forbidding the occurrence of
a pattern as a topological minor provides more expressive power than forbidding
arbitrary (finite) sets of patterns from occurring as sub-patterns.
\begin{proposition}
\label{prop:notstarlike}
If $P$ is a pattern that is not star-like, then 
$\ForbTM{P} \neq\ \ForbSP{\cal S}$
for all finite sets of patterns $\cal S$.
\end{proposition}
\begin{proof}
Let $P$ be a pattern that is not star-like, and let $P'$ be the negative pattern obtained
by removing all positive edges of $P$. 
Note that $\SP{P'}{P}$.

In any pattern, say that a part $U$ is \emph{distinguished} if two
negative edges share a single point in $U$ 
or if there are negative edges from $U$ to more than two other parts.

Since $P$ is not star-like, the negative pattern $P'$ must contains a
cycle of parts connected by negative edges,
or at least two distinguished parts. 

Hence, for any fixed $k$, by a sufficiently long 
sequence of subdivision operations, we can construct a subdivision $P''$ of $P'$ 
which either has a cycle of parts of length greater than $k$ or two
distinguished parts separated by a sequence of connected parts of length greater than $k$.
By adding positive edges, we can then convert $P''$ into a complete pattern 
of the form $\PI{I}$ for some CSP instance $I$.

Now for any fixed finite set of patterns $\cal S$ there will be a bound $k$ on  
the number of parts of any pattern in $\cal S$.
It follows that \ForbTM{P} cannot be defined by forbidding the sub-patterns in $\cal S$, 
since $I \notin \ForbTM{P}$ but no pattern in $\cal S$
can occur as a sub-pattern in $\PI{I}$.
\end{proof}

\section{Structural restrictions}
\label{sec:structural}

For any CSP instance $I=\tuple{V,D,C}$, the \emph{constraint graph} of $I$
is defined to be the graph $\tuple{V,E}$, where $E$ is the set of pairs $\{x,y\}$ 
for which the associated constraint $R_{xy}$ is non-trivial.
A number of tractable subproblems of the CSP have been defined by specifying
restrictions on the constraint graph; such restricted classes of instances are known as
\emph{structural classes}~\cite{Grohe07:otherside,DBLP:journals/jacm/Marx13}.

It is known that a structural class of binary CSP instances defined in this way is tractable if and only if 
every instance has a constraint graph of bounded treewidth~\cite[Theorem~5.1]{Grohe07:otherside}
(subject to the standard complexity-theoretic assumption that $\FPT \neq \WW$,
which we will assume throughout this section; we refer the reader to the
textbooks~\cite{Downey99:parametrized,Flum06:parametrized} for more details).
We show in this section that structural classes of this kind cannot be defined by forbidding 
the occurrence of a finite set of sub-patterns. 
However, they \emph{can} be defined by forbidding the occurrence of one or more
patterns as topological minors.

We will also use this characterisation of tractable structural classes
to show that a large class of negative patterns are topological minor tractable.

First we extend the notion of a constraint graph to arbitrary patterns.
\begin{definition}
\label{def:constraintgraphpattern}
For any pattern $P$, the \emph{constraint graph} of $P$, denoted $\CG{P}$,
is defined to be the graph $\tuple{V,E}$, where 
$V$ is the set of all parts of $P$, and $E$ is the set of pairs of parts $\{U,W\}$
such that there is a negative edge $(x,y) \in P$ with $x \in U$ and $y \in W$.
\end{definition}
For any binary CSP instance $I$, the constraint graph of $I$ defined above
is equal to $\CG{\PI{I}}$. For simplicity, this graph will usually be denoted by $G_I$.

Now we note the close link between 
our notion of a pattern occurring as a topological minor of another pattern
and the standard notion of a topological minor in a graph~\cite{Diestel10:graph}.
\begin{lemma}
\label{lem:graphtopminor}
For any graph $G$ and any pattern $P$,
$\TM{\PG{G}}{P}$ if and only if 
$G$ is a topological minor of the graph $\CG{P}$.
\end{lemma}
The simplest structural class of CSP instances of bounded treewidth is the class 
of instances whose constraint graph is {\em acyclic} (that is, has treewidth 1).
This class is known as the class of acyclic binary CSP instances
and was one of the first sub-problems of the CSP to be shown to be tractable~\cite{Freuder82:backtrack-free}.
We now show that this class can be characterised very simply by excluding 
the single pattern $\PG{C_3}$ shown in Figure~\ref{fig:acyclic}
from occurring as a topological minor.
\begin{proposition}
\label{prop:acyclic}
The class of acyclic binary CSP instances equals $\ForbTM{\PG{C_3}}$.
\end{proposition}
\begin{proof}
The class of acyclic graphs may be characterised as graphs which do not contain $C_3$
as a topological minor~\cite{Diestel10:graph}.
Hence, by Lemma~\ref{lem:graphtopminor} and Definition~\ref{def:constraintgraphpattern},
a binary CSP instance $I$ has an acyclic constraint graph if and only if 
it is not the case that $\TM{\PG{C_3}}{\PI{I}}$.
\end{proof}
Since the pattern $\PG{C_3}$ is not star-like (see Example~\ref{ex:starlike}),
it follows immediately from Proposition~\ref{prop:notstarlike} that 
acyclic CSP instances cannot be defined by any finite set of forbidden sub-patterns.
\begin{corollary}
\label{cor:acyclic}
The class of acyclic binary CSP instances is not equal to 
\ForbSP{\cal S} for any finite set of patterns $\cal S$.
\end{corollary}

Proposition~\ref{prop:acyclic}
can easily be extended to any of the tractable structural classes of binary CSP instances
defined by imposing any fixed bound on the treewidth of the constraint
graph~\cite{Freuder85:backtrack-bounded}, although in this case the set of of
forbidden patterns is explicitly known only for $k\leq 3$~\cite{Arnborg90:forbidden}.
\begin{theorem}
For any fixed $k\geq 1$, the class of binary CSP instances whose constraint
graph has treewidth at most $k$ equals \ForbTM{{\cal S}_k},
for some finite set of patterns ${\cal S}_k$.
\end{theorem}
\begin{proof}
The graph minor theorem~\cite{Robertson04:jctb} implies that for any fixed
$k\geq 1$ there is a finite set $O_k$ of graphs such that the class of graphs of
treewidth at most $k$ is precisely the class of graphs excluding all graphs from
the set $O_k$ as topological minors~\cite{Diestel10:graph}. (More precisely, the
graph minor theorem gives a finite set of minors as obstructions but this set
can be turned into a finite set of  
topological minors as obstructions in a standard way, 
see~\cite[Exercise 34, Chapter 12]{Diestel10:graph}.)
Consequently, by Lemma~\ref{lem:graphtopminor},
for any $k\geq 1$ the class of binary CSP instances with
constraint graphs of treewidth at most $k$ can be defined 
as $\ForbTM{{\cal S}_k}$ for the finite set of negative patterns ${\cal S}_k$
given by ${\cal S}_k = \{\PG{G} \mid G\in O_k\}$.
\end{proof}

In fact, we are able to show that many other patterns are topological-minor tractable 
using other standard results from graph theory. 
The following theorem characterises the topological-minor tractability of patterns of the
form $\PG{G}$, for all graphs $G$ of maximum degree three.
\begin{theorem}
Let $G$ be an arbitrary graph of maximum degree three. 
Then, \PG{G} is topological-minor tractable if and only if $G$ is planar
\emph{(}assuming $\FPT \neq \WW$\emph{)}.
\end{theorem}
\begin{proof}
One of the well-known results of Robertson and Seymour shows that the class of graphs obtained
by excluding $G$ as a minor has bounded treewidth if and only if $G$ is
planar~\cite{Robertson86:excluding} (see also~\cite[Theorem~12.4.3]{Diestel10:graph}). 
It is known that for a graph $G$ of maximum degree three and any graph $G'$, $G$
is a minor of $G'$ if and only if $G$ is a topological minor of
$G'$~\cite[Proposition~1.7.4~(ii)]{Diestel10:graph}. 
Thus, for a graph $G$ of maximum degree three, the class of graphs obtained by
excluding $G$ as a topological minor has bounded treewidth if and only if $G$ is
planar.
The theorem then follows from Lemma~\ref{lem:graphtopminor} and the fact that,
assuming $\FPT \neq \WW$, a structural class of binary CSP instances is
tractable if and only if the associated class of constraint graphs is of bounded
treewidth~\cite{Grohe07:otherside}.
\end{proof}
Unfortunately this result does not extend to graphs of higher degree, as the following 
example shows.
\begin{example}
Consider a star graph $G$ where the central vertex has degree 4.
Note that $G$ is planar.

In all subdivisions of $G$, the central vertex still has degree 4,
so it cannot occur as a topological minor in any graph of maximum degree three.
Hence, by Lemma~\ref{lem:graphtopminor}, $\PG{G}$ cannot occur as a topological minor in
any CSP instance whose constraint graph is a hexagonal grid.
Since the treewidth of the class of hexagonal grids is unbounded~\cite{Diestel10:graph},
this structural class of CSP instances is intractable, assuming $\FPT \neq \WW$,
by the results of~\cite{Grohe07:otherside}.
\end{example}

\section{Tractable classes that generalise acyclicity}
\label{sec:scheme}

In this section we will give several more examples of patterns that are 
topological-minor tractable. 
We conclude the section with Theorem~\ref{thm:P2}  
where we define several new tractable classes which properly extend 
the class of acyclic CSP instances discussed in Section~\ref{sec:structural}.

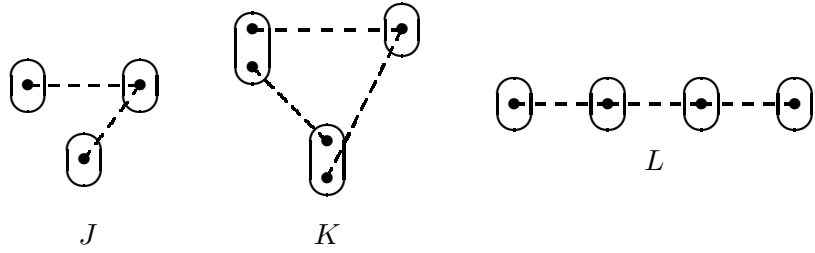
\begin{figure}
\thicklines \setlength{\unitlength}{0.7pt}
\centering

\begin{picture}(350,140)(50,-20)

\put(0,10){
\begin{picture}(90,90)(0,0)
\put(0,40){\usebox{\varone}} \put(60,40){\usebox{\varone}}
\put(30,0){\usebox{\varone}}
\dashline[50]{7}(40,20)(70,60) \dashline[50]{7}(10,60)(70,60)
\put(42,-20){\makebox(0,0){$J$}}
\end{picture}}

\put(120,10){ 
\begin{picture}(100,100)(0,0)
\put(0,60){\usebox{\vartwo}} \put(40,0){\usebox{\vartwo}}
\put(80,70){\usebox{\varone}} \dashline[50]{7}(10,70)(50,30)
\dashline[50]{7}(10,90)(90,90) \dashline[50]{7}(50,10)(90,90)
\put(50,-20){\makebox(0,0){$K$}}
\end{picture}}

\put(260,40){
\begin{picture}(130,100)(0,0)
\put(0,0){\usebox{\varone}} \put(50,0){\usebox{\varone}}
\put(100,0){\usebox{\varone}} \put(150,0){\usebox{\varone}}
\dashline[50]{7}(10,20)(60,20) \dashline[50]{7}(60,20)(110,20)
\dashline[50]{7}(110,20)(160,20) \put(85,-10){\makebox(0,0){$L$}}
\end{picture}}

\end{picture}

\caption{Three patterns which are topological-minor tractable.}
\label{fig:new}
\end{figure}

Consider the patterns shown in Figure~\ref{fig:new}. 
By Theorem~\ref{thm:pivotSPtractable}, 
$J$ is sub-pattern tractable and hence also topological-minor tractable,
by Lemma~\ref{lem:sp}.
However, the remaining patterns, $K$ and $L$ are more interesting.
\begin{theorem}
\label{thm:KisTMtractable}
The pattern $K$, shown in Figure~\ref{fig:new}, is 
sub-pattern $\NP$-complete but topological-minor tractable.
\end{theorem}
\begin{proof}
By Theorem~\ref{thm:pivotSPtractable}, $K$ is sub-pattern $\NP$-complete.

To show that $K$ is topological-minor tractable, 
consider an instance $I$ in which the pattern $K$ does not occur as a topological minor.
If the pattern $J$ from Figure~\ref{fig:new} does not occur as a sub-pattern in $\PI{I}$ 
then we are done since, as noted above, \ForbSP{J} is tractable 
and thus $I$ can be solved in polynomial time. 

On the other hand, if $J$ does occur as a sub-pattern in $\PI{I}$, then we will build a special tree 
decomposition $T$ of the constraint graph of $I$, where each node of $T$ is a subset of the
vertices of the constraint graph of $I$, and all non-leaf nodes of $T$ have size 1.

In more detail, let $G_I$ be the constraint graph of $I$.
Suppose the pattern $J$, shown in Figure~\ref{fig:new}, occurs as a sub-pattern in $\PI{I}$
on the three parts corresponding to the triple of variables $(x,y,z)$ in $I$, 
with $y$ being the variable at which the two negative edges meet. 
Since $K$ does not occur as a topological minor in $I$, it follows that there is no path from $x$ to $z$ 
in $G_I$ that does not pass through $y$. 
Hence $y$ is an articulation point of $G_I$.

Let $C_1,\ldots,C_k$ be the components of $G_I\setminus\{y\}$, 
and denote by $I_{C_i}$ the sub-instance of $I$ on the variables of $C_i \cup \{y\}$.
We form a tree decomposition of $G_I$ as follows:  
the root of $T$ is the subset containing just the variable $y$ and has $k$ children. 
If the pattern $J$ does not occur as a sub-pattern in $C_i \cup \{y\}$, 
then the $i$-th child of the root is a leaf node corresponding to the sub-instance $I_{C_i}$. 
Otherwise, if the pattern $J$ does occur as a sub-pattern in $C_i \cup \{y\}$, 
then we proceed in the same fashion and decompose $C_i$ into a sub-tree rooted at the $i$-th child. 

Since \ForbSP{J} is tractable, any sub-instance corresponding to a leaf of this tree decomposition
can be solved in polynomial time for each possible assignment to its unique articulation variable
which joins it to its parent node in the tree-decomposition. Hence in polynomial time
we can solve this sub-instance, eliminate the corresponding leaf, and possibly eliminate 
some values in the domain of this articulation variable. 
After eliminating all non-trivial leaf nodes in this way, 
the remaining sub-instance of $G_I$ is tree structured and hence can be solved in polynomial time.
\end{proof}

We will show in Theorem~\ref{thm:line} below 
that the pattern $L$ shown in Figure~\ref{fig:new} is also topological-minor tractable. 
In order to do so, we will extend the proof technique used in 
Theorem~\ref{thm:KisTMtractable} to a \emph{generic scheme} for proving 
topological-minor tractability of patterns.

To develop our generic scheme we need some standard results from graph theory.
If $S$ is a set of vertices of a graph $G$, we write $G[S]$ for the induced graph on $S$. 

A {\em tree decomposition} of a graph $G = (V,E)$ is a tree $T$, together with a subset $V_t$
of the vertices of $G$ for each node $t \in T$, such that $\bigcup_{t \in T} V_t = V$,
each edge $e \in E$ is contained in $V_t$ for some $t \in T$, 
and for any vertex $v \in V$ the set $\{t \mid v \in V_t\}$ is a connected sub-tree of $T$.
The \emph{torsos} of a tree decomposition $(T,(V_t)_{t\in T})$ of a graph $G$
are the graphs $H_t$, $t\in T$, obtained from $G[V_t]$ by adding all the
edges $\{x,y\}$ such that $x,y\in V_t\cap V_{t'}$ where $t'$ is any neighbour of $t$ in $T$.

A \emph{Tutte decomposition} of a graph $G$ is a tree decomposition
$(T,(V_t)_{t\in T})$ of $G$, where
$|V_t\cap V_{t'}|\leq 2$ for every pair of neighbours $t$ and $t'$ in $T$, 
and the torso of each node is either three-connected, or a cycle, or has at most 2 vertices.
It is known that every finite graph has a Tutte decomposition of this
kind~\cite{Tutte66:connectivity}, and that such a decomposition can be found in linear time~\cite{HopTar73}.

\begin{example}
Figure~\ref{fig:Tutte} shows a graph and a possible Tutte decomposition.
\end{example}

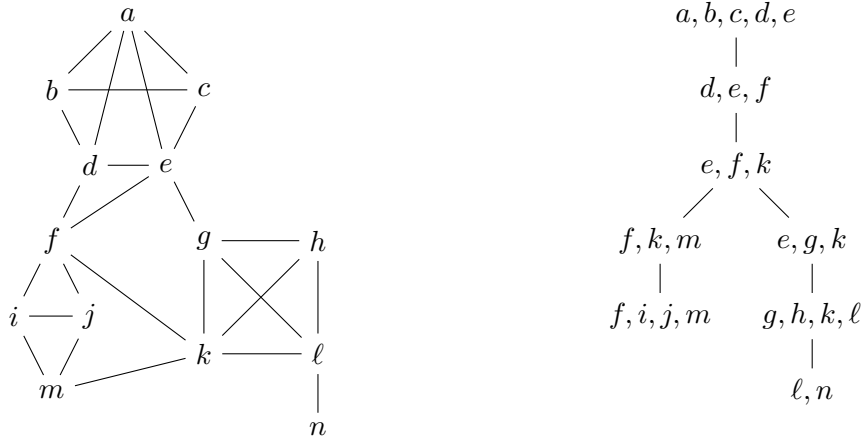
\begin{figure}[ht]
\begin{center}
\begin{tikzpicture}
  \node (a) at (0,5) {$a$};
  \node (b) at (-1,4) {$b$};
  \node (c) at (+1,4) {$c$};
  \node (d) at (-0.5,3) {$d$};
  \node (e) at (+0.5,3) {$e$};
  \node (f) at (-1,2) {$f$};
  \node (g) at (+1,2) {$g$};
  \node (h) at (+2.5,2) {$h$};
  \node (i) at (-1.5,1) {$i$};
  \node
   (j) at (-0.5,1) {$j$};
  \node (k) at (+1,0.5) {$k$};
  \node (l) at (+2.5,0.5) {$\ell$};
  \node (m) at (-1,0) {$m$};
  \node (n) at (+2.5,-0.5) {$n$};

  \node (t1) at (8,5) {$a,b,c,d,e$};
  \node (t2) at (8,4) {$d,e,f$};
  \node (t3) at (8,3) {$e,f,k$};
  \node (t4a) at (7,2) {$f,k,m$};
  \node (t4b) at (9,2) {$e,g,k$};
  \node (t5a) at (7,1) {$f,i,j,m$};
  \node (t5b) at (9,1) {$g,h,k,\ell$};
  \node (t6b) at (9,0) {$\ell,n$};

  \draw (a) -- (b); \draw (a) -- (c); \draw (a) -- (d); \draw (a) -- (e); \draw (b) -- (c); \draw (b) -- (d); \draw (c) -- (e); \draw (d) -- (e); \draw (d) -- (f); \draw (e) -- (f); \draw (e) -- (g); \draw (f) -- (i); \draw (f) -- (j); \draw (f) -- (k); \draw (g) -- (h); \draw (g) -- (k); \draw (g) -- (l); \draw (h) -- (k); \draw (h) -- (l); \draw (i) -- (j); \draw (i) -- (m); \draw (j) -- (m); \draw (k) -- (l); \draw (k) -- (m); \draw (l) -- (n);

  \draw (t1) -- (t2); \draw (t2) -- (t3); \draw (t3) -- (t4a); \draw (t3) -- (t4b); \draw (t4a) -- (t5a); \draw (t4b) -- (t5b); \draw (t5b) -- (t6b);
\end{tikzpicture}
\end{center}
\caption{A graph and its Tutte decomposition.}\label{fig:Tutte}
\end{figure}

To demonstrate topological-minor tractability for a pattern $P$ we proceed as follows.
Let $I$ be an instance in which $P$ does not occur as a topological minor and
let $G_I$ be its constraint graph. 
We denote by $n$ the number of variables in $I$ 
and by $d$ the maximum domain size of any variable in $I$.

Build a Tutte decomposition of $G_I$, and
consider any leaf node $s$ in this decomposition.
The subset of variables associated with node $s$ will be denoted $S$, 
and the variables associated with the remainder of the nodes of the tree decomposition 
after removing the leaf $s$
will be denoted by $T$. Note that $S$ and $T$ share at most 2 variables.
Let $I[S]$ be the sub-instance of $I$ on $S$
and $I[T]$ be the sub-instance of $I$ on $T$.
Suppose that the following two assumptions hold:
\begin{enumerate}
\item[(\textbf{A1})]
$I[S]$ can be solved and its solutions projected onto the
variables shared with $T$ 
in polynomial time;
the resulting reduced instance on $T$ will be denoted by $I'[T]$.
\item[(\textbf{A2})]
$P$ does not occur as a topological minor in $\PI{I'[T]}$.
\end{enumerate}
Then it follows that a recursive algorithm, which
at each step chooses some leaf $s$ of the decomposition, 
and then solves the associated sub-problem $I[S]$
to obtain the reduced instance $I'[T]$,
will solve the original instance using a
polynomial (in $n$ and $d$) number of calls to the polynomial-time algorithm from
(\textbf{A1}).

In the proofs below we will omit the simple cases 
where $S$ and $T$ share only 1 variable,  
or $S$ contains at most 3 vertices, 
or the torso of $S$ is a cycle (and hence has treewidth 2
and is solvable in polynomial time).
Hence we will assume that the torso of $S$ contains more than three vertices 
and is three-connected.

Finally, note that if $S$ and $T$ share the variables $\{u,v\}$,
then we have the following:
\begin{itemize}
\item Any path in $G_I$ from a vertex in $S$ to a vertex in $T$ 
must pass through $u$ or $v$;
\item There must exist some path from $u$ to $v$ in $G_I[T]$,
which we will denote $path_{T}(u,v)$.
\end{itemize}

We now use this generic scheme to prove the tractability of pattern $L$ from
Figure~\ref{fig:new}.
\begin{theorem} 
\label{thm:line}
The pattern $L$, shown in Figure~\ref{fig:new}, is sub-pattern $\NP$-complete 
but topological-minor tractable.
\end{theorem}
\begin{proof}
By Theorem~\ref{thm:pivotSPtractable}, $L$ is sub-pattern $\NP$-complete.

To establish topological-minor tractability using the generic scheme described above 
we only only need to establish the two assumptions.

(\textbf{A1})\quad
Let $J$ be the pattern consisting of two intersecting negative edges, shown in Figure~\ref{fig:new}.
Suppose that $J$ occurs in $\PI{I[S]}$ as a sub-pattern on two \emph{disjoint} triples
of variables $(x,y,z)$ and $(x',y',z')$ in $I[S]$.
As explained above for the generic scheme, we can assume
that the torso of $S$ is 3-connected. 
It follows by Menger's theorem~\cite{Dirac66:Menger} that there are
three disjoint paths from $x$ to $x'$ in the torso of $S$.
There must be one of these paths, $\pi$, which does
not pass through $y$ or $y'$. 
We claim that there must be a subpath $\sigma$ of $\pi$ which begins at $x$ or
  $z$ and ends at $x'$ or $z'$ and which does not pass through any other
  variables in $\{x,y,z,x',y',z'\}$. To prove the claim first note that if
  $\pi$ does not pass through $z$ and $z'$ then $\pi$ satisfies the claim. If 
  $z$ appears on $\pi$ but $z'$ does not appear on $\pi$ then the subpath
  $\sigma$ of $\pi$ from $z$ to $x'$ satisfies
  the claim. A similar argument works for the case when $z'$ appears on $\pi$ but
  $z$ does not. If both $z$ and $z'$ appear on $\pi$ 
  then we have a subpath of $\pi$ from $z$ to $z'$.
Without loss of generality, suppose that $\sigma$ joins $x$ to $x'$. But then $L$ occurs
as a topological minor on the extended path $\sigma^{+}$
given by $z \rightarrow y \rightarrow x, \sigma,
x' \rightarrow y' \rightarrow z'$.

But this implies that $L$ occurs as a
topological minor in $\PI{I}$, since if $\sigma^{+}$ passes
by the edge $\{u,v\}$ in the torso of $S$, this edge can be replaced
by $path_T(u,v)$ which is a path from $u$ to $v$ in $T$, whose
existence was noted in the discussion above.  
Since this contradicts
our initial assumption, we can deduce that
$J$ does not occur in $\PI{I[S]}$ as a sub-pattern on two \emph{disjoint} triples.

We can therefore deduce that all pairs of triples of variables
$(x,y,z)$, $(x',y',z')$ for which $J$ occurs as a sub-pattern in
$\PI{I[S]}$ intersect, i.e., $\{x,y,z\} \cap \{x',y',z'\} \neq \emptyset$.
Now, consider an arbitrary
triple of variables $(x,y,z)$ on which $J$ occurs as a sub-pattern.
It follows that the instance which results after any instantiation (and removal) of
the three variables $x,y,z$ contains no occurrence of $J$ as a sub-pattern, since for
each triple of variables $(x',y',z')$ on which $J$ occurs in $I[S]$, at least
one of its variables has been eliminated by instantiation.

Thus, after instantiation of at most three variables, $\PI{I[S]}$
does not contain $J$ as a sub-pattern. This also holds for
any version of $I[S]$ obtained by instantiating the variables $u,v$.
As noted above, \ForbSP{J} is tractable.
We can therefore determine in polynomial time which instantiations of
$u,v$ can be extended to a solution of $I[S]$. 
We remove the pair $(p,q)$ from $R_{uv}$ in $I$
whenever the assignment of $p$ to $u$ and $q$ to $v$ 
cannot be extended to a solution to $I[S]$.
Finally, we delete all variables in $S$ from $I$ apart from $u$ and $v$. 
Proceeding in this way we construct $I'[T]$ in polynomial time as required.

(\textbf{A2}):\quad
Suppose, for a contradiction, that we introduce some occurrence of the pattern $L$ as a
topological minor in $\PI{I'[T]}$ when reducing $I$ to $I'[T]$. 
This occurrence of $L$ must use a newly-introduced edge in $I'[T]$.
During the reduction from $I$ to $I'[T]$, 
we can introduce negative (but not positive) edges in $\PI{I'[T]}$
between the parts corresponding to $u$ and $v$. 
Suppose that a negative edge $(p,q)$ is introduced by the reduction 
from $I$ to $I'[T]$. 
This can only be the case if there was a path $\pi = (u,w_1,\ldots,w_t,v)$ 
in the constraint graph $G_I[S]$ 
and hence a sequence of negative edges between the corresponding parts in $\PI{I[S]}$ 
linking $p$ to $q$.
This means that we can replace the newly-introduced edge 
in the occurrence of $L$ in $\PI{I'[T]}$ by a sequence of negative edges
so that $L$ occurs as a topological minor in $\PI{I}$ for the original instance $I$. 
This contradiction shows that we cannot introduce $L$ as a topological
minor in $\PI{I'[T]}$ when reducing $I$ to $I'[T]$.

Hence we have established both assumptions, so the result follows by our generic proof scheme.
Note that the number of instances of \ForbSP{J} that need to be solved
is $O(nd^5)$.
\end{proof}

As our final result in this section we show how the well-known tractable class
of acyclic instances can be generalised to obtain larger tractable classes
defined by forbidding the occurrence of certain patterns as topological minors. 
The main tool we use will again be the generic scheme based on Tutte decompositions
described above.
\begin{theorem} \label{thm:P2}
Let $P_0$ be any sub-pattern tractable pattern with three parts,
$U_1,U_2,U_3$ where there is at most one negative edge between $U_1$
and $U_2$, and between $U_2$ and $U_3$, and no edges between $U_1$ and $U_3$.

Let $P$ be a pattern with four parts $U_1,U_2,U_3,U_4$ obtained by extending $P_0$ as follows.
The pattern $P$ has six new points $p_1,p_2 \in U_1$, $q_1,q_2 \in U_4$, 
and $r_1,r_2 \in U_3$, together with three new
negative edges $\{p_1,r_1\}$, $\{p_2,q_1\}$, $\{q_2,r_2\}$ (see Figure~\ref{fig:newfromthm}).
Any such $P$ is topological-minor tractable.
\end{theorem}
\thicklines \setlength{\unitlength}{0.7pt}
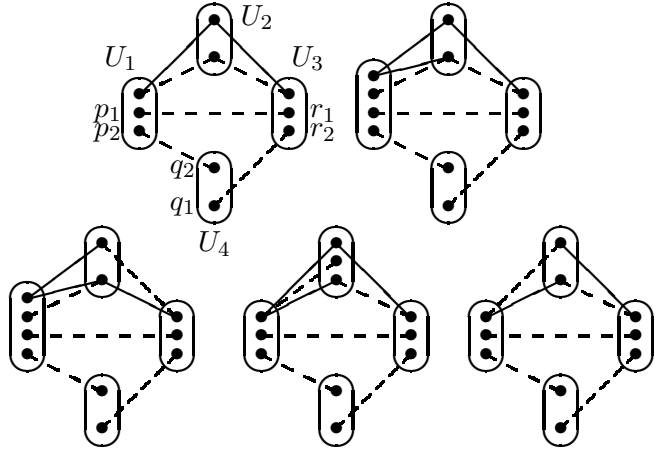
\begin{figure}[ht]
\centering

\begin{picture}(360,250)(0,0)

\put(60,130){\begin{picture}(100,120)(0,0)
\put(0,40){\usebox{\varthree}} \put(40,0){\usebox{\vartwo}}
\put(80,40){\usebox{\varthree}} \put(40,80){\usebox{\vartwo}}
\put(50,110){\line(1,-1){40}} \put(10,70){\line(1,1){40}}
\dashline[50]{7}(10,70)(50,90) \dashline[50]{7}(50,90)(90,70)
\dashline[50]{7}(10,60)(90,60) \dashline[50]{7}(10,50)(50,30)
\dashline[50]{7}(50,10)(90,50)
\put(0,90){\makebox(0,0){$U_1$}}
\put(100,90){\makebox(0,0){$U_3$}}
\put(73,112){\makebox(0,0){$U_2$}}
\put(50,-10){\makebox(0,0){$U_4$}}
\put(-7,50){\makebox(0,0){$p_2$}}
\put(-7,60){\makebox(0,0){$p_1$}}
\put(108,50){\makebox(0,0){$r_2$}}
\put(108,60){\makebox(0,0){$r_1$}}
\put(33,30){\makebox(0,0){$q_2$}}
\put(33,10){\makebox(0,0){$q_1$}}

\end{picture}}

\put(180,130){
\begin{picture}(100,120)(0,0)
\put(0,40){\usebox{\varfour}} \put(40,0){\usebox{\vartwo}}
\put(80,40){\usebox{\varthree}} \put(40,80){\usebox{\vartwo}}
\put(50,110){\line(1,-1){40}} \put(10,80){\line(4,1){40}}
\put(10,80){\line(4,3){40}} \dashline[50]{7}(10,70)(50,90)
\dashline[50]{7}(50,90)(90,70) \dashline[50]{7}(10,60)(90,60)
\dashline[50]{7}(10,50)(50,30) \dashline[50]{7}(50,10)(90,50)
\end{picture}}

\put(0,10){\begin{picture}(100,120)(0,0)
\put(0,40){\usebox{\varfour}} \put(40,0){\usebox{\vartwo}}
\put(80,40){\usebox{\varthree}} \put(40,80){\usebox{\vartwo}}
\put(50,90){\line(2,-1){40}} \put(10,80){\line(4,1){40}}
\put(10,80){\line(4,3){40}} \dashline[50]{7}(10,70)(50,90)
\dashline[50]{7}(50,110)(90,70) \dashline[50]{7}(10,60)(90,60)
\dashline[50]{7}(10,50)(50,30) \dashline[50]{7}(50,10)(90,50)
\end{picture}}

\put(120,10){
\begin{picture}(100,120)(0,0)
\put(0,40){\usebox{\varthree}} \put(40,0){\usebox{\vartwo}}
\put(80,40){\usebox{\varthree}} \put(40,80){\usebox{\varthree}}
\put(10,70){\line(1,1){40}} \put(50,110){\line(1,-1){40}}
\put(10,70){\line(2,1){40}} \dashline[50]{7}(10,70)(50,100)
\dashline[50]{7}(50,90)(90,70) \dashline[50]{7}(10,60)(90,60)
\dashline[50]{7}(10,50)(50,30) \dashline[50]{7}(50,10)(90,50)
\end{picture}}

\put(240,10){
\begin{picture}(100,120)(0,0)
\put(0,40){\usebox{\varthree}} \put(40,0){\usebox{\vartwo}}
\put(80,40){\usebox{\varthree}} \put(40,80){\usebox{\vartwo}}
\put(50,110){\line(1,-1){40}} \put(10,70){\line(2,1){40}}
\dashline[50]{7}(10,70)(50,110) \dashline[50]{7}(50,90)(90,70)
\dashline[50]{7}(10,60)(90,60) \dashline[50]{7}(10,50)(50,30)
\dashline[50]{7}(50,10)(90,50)
\end{picture}}

\end{picture}

\caption{Topological-minor tractable patterns derived from sub-pattern tractable
patterns.}
\label{fig:newfromthm}
\end{figure}
\begin{proof}
The proof uses the generic scheme described in this section, so we
only need to establish the two assumptions.

(\textbf{A1})\quad
Suppose first that $P_{0}$ occurs as a sub-pattern in $\PI{I[S]}$ on the
triple of variables $(x,y,z)$. As explained above, when using the generic scheme
we will assume that the torso of $S$ is three-connected. Then, by Menger's theorem
there are three disjoint paths $\pi_1, \pi_2, \pi_3$ from $x$ to $z$
in the torso of $S$. 
Hence there must be two of these paths, say $\pi_1$ and $\pi_2$,
which do not pass through $y$. But this implies that $P$ occurs as a
topological minor in $\PI{I}$, since if either $\pi_1$ or $\pi_2$ passes
through the edge $\{u,v\}$ in the torso of $S$, this edge can be replaced
by $path_T(u,v)$ which is a path from $u$ to $v$ in $G_I[T]$, whose
existence was shown in the discussion of the generic scheme above. 
Since this contradicts our initial assumption, we can assume that $P_0$ does
{\em not} occur as a sub-pattern in $\PI{I[S]}$. This also holds for any
sub-problem of $I[S]$ obtained by instantiating the variables $u,v$.
Therefore, by the sub-pattern tractability of $P_{0}$, we can determine
in polynomial time which instantiations of $u,v$ can be extended to a
solution of $I[S]$. 
We remove the pair $(p,q)$ from $R_{uv}$ in $I$
whenever the assignment of $p$ to $u$ and $q$ to $v$ 
cannot be extended to a solution to $I[S]$.
Finally, we delete all variables in $S$ from $I$ except for $u$ and $v$. 
Proceeding in this way we construct $I'[T]$ in polynomial time, as required.

(\textbf{A2})\quad
Suppose, for a
contradiction, that we introduce the pattern $P$ as a topological
minor of $\PI{I'[T]}$ when reducing $I$ to $I'[T]$. 
This occurrence of $P$ must use a newly-introduced negative edge.
Observe that, by definition, $P$ contains at
most one negative edge between any two parts. 
Suppose that a
negative edge $(p,q)$ is introduced by the reduction from $I$ to $I'[T]$. 
This can only be the case if there was a path 
$\pi = (u,w_1,\ldots,w_t,v)$ in the constraint graph $G_I[S]$ 
and hence a sequence of negative edges between the corresponding parts in $\PI{I[S]}$ 
linking $p$ to $q$.
Furthermore, in $I'[T]$, if there is a positive edge $(p',q')$
between the parts corresponding to $u$ and $v$
then there is necessarily a solution to $I[S]$ including the assignments 
$p'$ to $u$ and $q'$ to $v$ (and hence a solution on the subinstance $I[\pi]$ of 
$I[S]$ on the path $\pi = (u,w_1,\ldots,w_t,v)$ in $I[S]$). 
This means that we can replace the edge $(p,q)$ in the occurrence of 
$P$ in $I'[T]$ by a sequence of negative edges 
so that $P$ occurs as a topological minor in $\PI{I}$ for the original instance $I$.
This contradiction shows that we cannot introduce an occurrence of $P$ as a
topological minor in $\PI{I'[T]}$ when reducing $I$ to $I'[T]$.

Hence we have established both assumptions, so the result follows by our generic proof scheme.
Note that the number of instances of $\ForbSP{P_0}$ that need to be solved is $O(nd^2)$.
\end{proof}

By~\cite[Theorem 1]{Cooper15:dam}, all sub-pattern tractable patterns $P_0$ satisfying the conditions of 
Theorem~\ref{thm:P2} can be reduced to sub-patterns of one of five specific patterns.
Extending each of these to a pattern $P$ as described in Theorem~\ref{thm:P2}
gives the five topological-minor tractable patterns shown in Figure~\ref{fig:newfromthm}. 
For each of these patterns $P$,
the pattern shown in Figure~\ref{fig:acyclic} occurs as a sub-pattern
and hence as a topological minor of $P$. Thus, by the transitivity of
occurrence as a topological minor, each tractable class
\ForbTM{P} necessarily contains all acyclic binary CSP instances.

\section{Detection of topological minors}\label{sec:detecting}

For every fixed undirected graph $H$, there is an $O(n^3)$ time
algorithm that tests, given a graph $G$ with $n$ vertices, if $H$ is a topological
minor of $G$~\cite{DBLP:conf/stoc/GroheKMW11}.

However, for detecting topological minors in {\em patterns} the situation is
different. 
Characterising all patterns $P$ for which it is possible to decide in polynomial time whether
$P$ occurs as a topological minor in a given pattern $P'$
remains an open problem. However, we have the following partial results.

By Lemma~\ref{lem:graphtopminor}, deciding whether a negative pattern 
of the form $\PG{G}$ for some graph $G$ occurs as a topological
minor in a pattern $P'$ amounts to detecting whether  
$G$ is a topological minor of the constraint graph of $P'$, and hence
can be achieved in polynomial time~\cite{DBLP:conf/stoc/GroheKMW11}.
By Proposition~\ref{prop:star},
deciding whether a star-like negative pattern occurs as a topological minor in an instance 
can also be achieved in
polynomial time because this is equivalent to deciding whether it
occurs as a sub-pattern, which is achievable in polynomial time by exhaustive
search.
\begin{proposition}
For each of the patterns $J$, $K$ or $L$ shown in Figure~\ref{fig:new}, deciding whether that 
pattern occurs as a topological minor in a given instance $I$ can be done in polynomial
time.
\end{proposition}
\begin{proof}
The pattern $J$ shown in Figure~\ref{fig:new} is star-like, and hence the result
follows from the observation just made.
For the pattern $K$ shown in Figure~\ref{fig:new} it is sufficient to discover 
by exhaustive search all occurrences of $J$ as a sub-pattern of $\PI{I}$ 
on the three parts corresponding to the triple of variables $(x,y,z)$ in $I$, 
with $y$ being the variable at which the two negative edges meet, 
and then check for each one whether $x$ and $z$ are connected in $\CG{I} \setminus y$.
 
For the pattern $L$ shown in Figure~\ref{fig:new} it is sufficient to consider all 
pairs of occurrences of $J$ as a sub-pattern of $\PI{I}$ on parts corresponding to 
$(x,y,z)$ and $(x',y',z')$
(where the negative edges meet in parts $y$ and $y'$).
We can then check that either $(y,z)$ and $(x',y')$ coincide, or $z$ and $x'$ coincide,
or $z$ and $x'$ are connected by a path in $\CG{I}$
that does not pass through any of the parts $x,y,y',z'$.
\end{proof}

For each of the patterns shown in Figure~\ref{fig:newfromthm} the complexity of deciding whether it
occurs as a topological minor in a given instance $I$ is currently unknown.
However, in polynomial time we can build a Tutte decomposition for $I$ and decide whether
each of the sub-problems associated with its nodes are members of 
\ForbSP{P_0} for the appropriate pattern $P_0$, and this is the only condition required 
to solve $I$ in polynomial time using the algorithm described in the proof of Theorem~\ref{thm:P2}. 

Our next result shows that for some patterns (such as the 4-part pattern $M$ shown in
Figure~\ref{fig:patternPXX}), it is $\coNP$-complete to determine whether
the pattern occurs as a topological minor in an arbitrary given pattern.

\begin{figure}[ht]
\thicklines \setlength{\unitlength}{0.7pt}
\centering

\begin{picture}(300,80)(0,0)

\put(80,0){
\begin{picture}(140,50)(0,-20)
\put(0,0){\usebox{\vartwo}} \put(40,0){\usebox{\vartwo}}
\put(80,0){\usebox{\vartwo}} \put(120,0){\usebox{\vartwo}}
\put(10,30){\line(1,0){40}} \put(10,30){\line(2,-1){40}}
\put(10,10){\line(2,1){40}} \put(50,30){\line(1,0){40}}
\put(90,30){\line(1,0){40}} \put(90,30){\line(2,-1){40}}
\put(90,10){\line(2,1){40}} \dashline[50]{7}(10,10)(50,10)
\dashline[50]{7}(90,10)(130,10)
\put(75,-20){\makebox(0,0){$M$}}
\end{picture}
}
\end{picture}
\caption{A pattern that is coNP-complete to detect as a topological minor.}
\label{fig:patternPXX}
\end{figure}
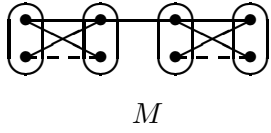

\begin{figure}[ht]
\thicklines \setlength{\unitlength}{0.7pt}
\centering

\begin{picture}(300,380)(0,-40)

\put(0,170){
\begin{picture}(300,130)(-20,-25)
\put(0,20){\usebox{\varone}} \put(40,0){\usebox{\vartwo}}
\put(80,0){\usebox{\vartwo}}
\put(140,0){\usebox{\vartwo}}
\put(200,0){\usebox{\vartwo}}
\put(40,50){\usebox{\vartwo}} \put(80,50){\usebox{\vartwo}}
\put(140,50){\usebox{\vartwo}}
\put(200,50){\usebox{\vartwo}} \put(240,20){\usebox{\varone}}
\put(10,40){\line(2,1){40}} \put(10,40){\line(4,-3){40}}

\put(50,60){\line(1,0){60}}
\thinlines \dottedline{2}(110,60)(125,60) \thicklines
\put(125,60){\line(1,0){45}}
\thinlines \dottedline{2}(170,60)(185,60) \thicklines
\put(185,60){\line(1,0){25}}

\put(50,10){\line(1,0){60}}
\thinlines \dottedline{2}(110,10)(125,10) \thicklines
\put(125,10){\line(1,0){45}}
\thinlines \dottedline{2}(170,10)(185,10) \thicklines
\put(185,10){\line(1,0){25}}

\put(210,60){\line(2,-1){40}} \put(210,10){\line(4,3){40}}
\put(-13,30){\makebox(0,0){$p_{i-1}$}}
\put(270,30){\makebox(0,0){$p_{i}$}}
\put(40,99){\makebox(0,0){$v_{i1}$}}
\put(80,99){\makebox(0,0){$v_{i2}$}}
\put(190,99){\makebox(0,0){$v_{im}$}}
\put(40,-10){\makebox(0,0){$\overline{v}_{i1}$}}
\put(80,-10){\makebox(0,0){$\overline{v}_{i2}$}}
\put(190,-10){\makebox(0,0){$\overline{v}_{im}$}}
\put(130,-28){\makebox(0,0){(a)}}
\end{picture}}

\put(0,0){
\begin{picture}(150,140)(-30,0)
\put(0,50){\usebox{\varone}} \put(40,0){\usebox{\vartwo}}
\put(40,50){\usebox{\vartwo}} \put(40,100){\usebox{\vartwo}}
\put(80,50){\usebox{\varone}} \put(10,70){\line(1,-1){40}}
\put(10,70){\line(4,1){40}} \put(10,70){\line(2,3){40}}
\put(50,30){\line(1,1){40}} \put(50,80){\line(4,-1){40}}
\put(50,130){\line(2,-3){40}}
\put(-25,60){\makebox(0,0){$p_{n+r-1}$}}
\put(118,60){\makebox(0,0){$p_{n+r}$}}
\put(30,120){\makebox(0,0){$\overline{v}_{jr}$}}
\put(31,65){\makebox(0,0){$v_{kr}$}}
\put(30,20){\makebox(0,0){$\overline{v}_{\ell r}$}}
\put(50,-20){\makebox(0,0){(b)}}
\end{picture}}

\put(200,85){
\begin{picture}(90,50)(-20,-10)
\put(0,0){\usebox{\vartwo}} \put(40,0){\usebox{\vartwo}}
\put(10,30){\line(1,0){40}} \put(10,30){\line(2,-1){40}}
\put(10,10){\line(2,1){40}}
\put(74,20){\makebox(0,0){$p_{0}$}}
\put(-10,20){\makebox(0,0){$u$}}
\put(30,-20){\makebox(0,0){(c)}}
\end{picture}}

\put(200,0){
\begin{picture}(90,50)(-20,0)
\put(0,0){\usebox{\vartwo}} \put(40,0){\usebox{\vartwo}}
\put(10,30){\line(1,0){40}} \put(10,30){\line(2,-1){40}}
\put(10,10){\line(2,1){40}}
\put(72,20){\makebox(0,0){$w$}}
\put(-22,20){\makebox(0,0){$p_{n+m}$}}
\put(30,-20){\makebox(0,0){(d)}}
\end{picture}}

\end{picture}

\caption{The building blocks for the CSP
instance $I$ constructed in the proof of Theorem~\ref{thm:PXX}.} 
\label{fig:npc}
\end{figure}
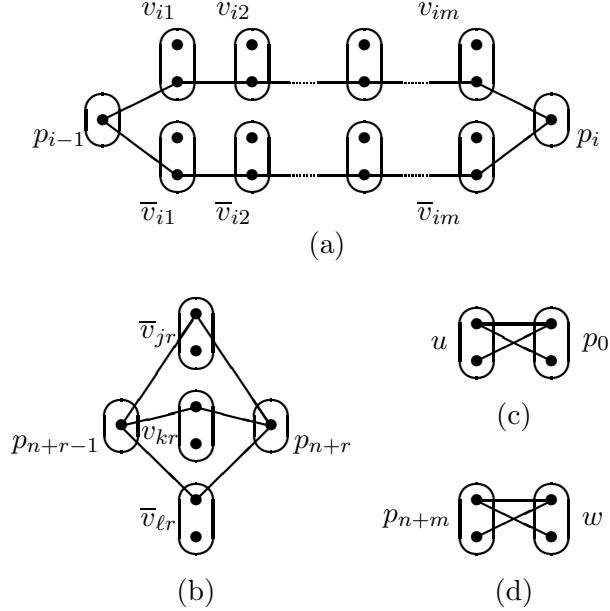

\begin{theorem}
\label{thm:PXX}
The problem of deciding $I \in \ForbTM{M}$ is $\coNP$-complete.
\end{theorem}

\begin{proof}
The problem is clearly in $\coNP$, so it suffices to give a reduction
from 3-SAT to the complement of the problem of deciding $I \in \ForbTM{M}$. 

Let $I_{SAT}$ be an instance of 3-SAT with variables
$x_1,\ldots,x_n$ and clauses $C_1,\ldots,C_m$. We will create a binary CSP
instance $I$ with variables  $\{u,w\} \cup \{p_i\mid i=0\dots n+m\}
\cup \{v_{ir},\overline{v}_{ir}\mid i=1\dots n, r=1\dots m\}$, such
that determining whether $\TM{M}{\PI{I}}$ is equivalent to deciding whether 
$I_{SAT}$ has a solution.  The instance $I$ that we create will be Boolean 
in the sense that all variables will have domain size at most two.  
(In fact all the variables $p_i$, except for $p_0$ and $p_{n+m}$, will have single-valued domains.)

Consider the patterns shown in Figure~\ref{fig:npc}, where each part is labelled
with a variable of $I$.   
Using these patterns we build a complete pattern corresponding to the instance $I$, as follows:
\begin{itemize}
\item
For each variable $x_i$ in $I_{SAT}$ we include a pattern $P_{x_i}$
of the form shown in Figure~\ref{fig:npc}(a).
\item
For each clause $C_r$ in $I_{SAT}$ we include a pattern $P_{C_r}$ 
of the form shown in Figure~\ref{fig:npc}(b), 
where the choice of variables for the three central parts 
depends on the literals in the clause $C_r$ in the following way: 
variable $v_{ir}$ corresponds to
$\neg x_i$ occurring in clause $C_r$ and variable $\overline{v}_{ir}$
corresponds to $x_i$ occurring in clause $C_r$.  
That is, the example shown in Figure~\ref{fig:npc}(b) would correspond to the
clause $x_j \vee \neg x_k \vee x_\ell$. 
\item
We also include the pattern shown in Figure~\ref{fig:npc}(c) and the pattern
shown in Figure~\ref{fig:npc}(d);
\item Finally, we complete the resulting pattern to obtain $\PI{I}$ 
by adding negative edges between all pairs of points in distinct parts that are not already 
directly connected by a positive or negative edge.
\end{itemize}

The only pairs of parts in $\PI{I}$ that are connected by more than one
positive edge are $\{u,p_0\}$ and $\{p_{n+m},w\}$.
So, if $M$ occurs as a topological minor in $\PI{I}$, then the 
points of $M$ must map injectively to these two pairs of parts.
Therefore, deciding whether $M$ occurs as a topological minor in $\PI{I}$ is
equivalent to deciding whether there is a path $\pi$ of positive
edges from $p_0$ to $p_{n+m}$ in $\PI{I}$
which passes through each part at most once. 

Any such path $\pi$ must pass through the points $p_0,p_1,\ldots,p_{n+m}$ in this order,
because the positive edges in $P_{x_i}$ ($1 \leq i \leq n$) 
use different points in each part (shown as the bottom of
the two points in Figure~\ref{fig:npc}) 
from the positive edges in $P_{C_r}$ ($1 \leq r \leq m$)
(which use the top points), so there are no short-cuts.

If such a path $\pi$ exists, then for each variable $x_i$ of
$I_{SAT}$, the path $\pi$ must select in $P_{x_i}$ either the upper path through
variables $v_{ir}$ ($r=1,\ldots,m$) or the lower path through
variables $\overline{v}_{ir}$ ($r=1,\ldots,m$). 
Thus $\pi$ selects a truth value for each variable $x_i$: TRUE if $\pi$ follows
the upper of these two paths, FALSE otherwise. 

Moreover, for each clause $C_r$ in $I_{SAT}$ the path $\pi$
must pass from $p_{n+r-1}$ to $p_{n+r}$ by one of the three paths in $P_{C_r}$ 
without passing through parts that have
been already used by $\pi$. Thus, for $\pi$ to exist it must have
already assigned TRUE to one of the literals of the clause $C_r$.

It follows that $M$ occurs as a topological minor of $\PI{I}$
if and only if $\PI{I}$ has an appropriate path of positive edges, which occurs if
and only if $I_{SAT}$ is satisfiable.
\end{proof}

The instance $I$ in the proof of Theorem~\ref{thm:PXX} is clearly
inconsistent since there are some constraint relations which are
empty. An instance is said to be \emph{globally consistent} if each
variable-value assignment $\tuple{v_i,a}$ can be extended to a
solution. We now give another example of a pattern which is
$\coNP$-complete to detect as a topological minor even in
globally-consistent instances. 

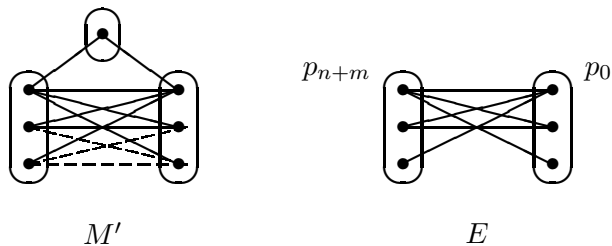
\begin{figure}[ht]
\centering
\thicklines \setlength{\unitlength}{0.7pt}
\newsavebox{\varthreebig}
\savebox{\varthreebig}(20,60){
\begin{picture}(20,60)(2,0)
\put(10,30){\oval(20,60)} \put(10,10){\makebox(0,0){$\bullet$}}
\put(10,30){\makebox(0,0){$\bullet$}} \put(10,50){\makebox(0,0){$\bullet$}}
\end{picture}}

\begin{picture}(350,130)(0,0)
\put(0,0){
\begin{picture}(150,130)(0,10)
\put(20,40){\usebox{\varthreebig}} \put(100,40){\usebox{\varthreebig}}
\put(60,100){\usebox{\varone}}
\put(30,90){\line(4,3){40}}  \put(110,90){\line(-4,3){40}}
\put(30,90){\line(1,0){80}} \put(30,70){\line(1,0){80}}
\put(30,90){\line(4,-1){80}} \put(30,70){\line(4,1){80}}
\put(30,90){\line(2,-1){80}} \put(30,50){\line(2,1){80}}
\thinlines
 \dashline[50]{7}(30,50)(115,50)
 \dashline[50]{7}(30,70)(115,50) \dashline[50]{7}(30,50)(115,70)
\thicklines \put(70,13){\makebox(0,0){$M'$}}
\end{picture}}

\put(200,0){
\begin{picture}(150,130)(0,10)
\put(20,40){\usebox{\varthreebig}} \put(100,40){\usebox{\varthreebig}}
\put(30,90){\line(1,0){80}} \put(30,70){\line(1,0){80}}
\put(30,90){\line(4,-1){80}} \put(30,70){\line(4,1){80}}
\put(30,90){\line(2,-1){80}} \put(30,50){\line(2,1){80}}
\put(135,100){\makebox(0,0){$p_{0}$}}
\put(-5,100){\makebox(0,0){$p_{n+m}$}}
\put(70,13){\makebox(0,0){$E$}}
\end{picture}}
\end{picture}

\caption{The pattern $M'$ and one of the building blocks for the
globally-consistent instance $I'$ in which detecting it is $\coNP$-complete.}
\label{fig:patternPX}
\end{figure}

\begin{theorem}
\label{thm:PX}
The problem of deciding $I \in \ForbTM{M'}$ 
for globally-consistent instances $I$ is $\coNP$-complete.
\end{theorem}

\begin{proof}
We use a very similar construction to the one used in the proof of Theorem~\ref{thm:PXX}.
Let $I$ be the instance constructed in that proof. Let $I'$ be identical to $I$ except that:
\begin{itemize}
\item we replace the sub-instances obtained from the patterns
shown in Figure~\ref{fig:npc}(c) and Figure~\ref{fig:npc}(d)
with a single sub-instance obtained from the pattern  $E$
shown in Figure~\ref{fig:patternPX};
\item for each variable-value assignment $\tuple{v,a}$ of $I$, we create a solution
which is an extension of $\tuple{v,a}$,
by adding a new value $b(v,a,v')$ to the domain of each variable $v' \neq v$
which is compatible with $\tuple{v,a}$
and with all such values $b(v,a,v'')$ ($v'' \notin \{v,v'\}$),
but incompatible with all other variable-value assignments.
\end{itemize}
By construction, $I'$ is clearly globally-consistent. 
If $M'$ occurs as a topological minor of $\PI{I'}$, then the points of $M'$
must map injectively to the points of $E$,
and so again the question is whether there is a path 
(of length greater than 1) of positive edges
linking $p_0$ to $p_{n+m}$. As in the proof of Theorem~\ref{thm:PXX},
this path exists if and only if the instance $I_{SAT}$ is
satisfiable. Hence, the decision problem 
$I \in \ForbTM{P_{X}}$ for globally-consistent instances $I$ is $\coNP$-complete. 
\end{proof}

Theorems~\ref{thm:PXX} and \ref{thm:PX} show that not all classes defined by forbidding topological minors
can be recognized in polynomial time.
Certain uses of tractable classes require polynomial-time
recognition: in particular, the automatic recognition and resolution
of easy instances within general-purpose solvers. On the other hand,
polynomial-time recognition of a tractable class $\mathcal{C}$ is not
required for the construction of a polynomial-time solvable
relaxation in $\mathcal{C}$, nor in the proof (by a human being) that
a subproblem of CSP encountered in practice falls in $\mathcal{C}$.

\section{Augmented patterns}\label{sec:augmented}

For some CSP instances we have extra information such as an ordering on the variables
or on the domains (or both). 
In this section we introduce the idea of adding an additional relation to a pattern 
to allow us to capture information of this kind.
A pattern $P$, together with an additional relation on the points of $P$ 
will be called an \emph{augmented pattern}. 
We will demonstrate that augmented patterns can be used to define new hybrid tractable classes
that extend those described in earlier sections.

\begin{definition}
\label{def:extendedpatt}
An \emph{augmented pattern} is a pair $\tuple{P,R}$ where 
$P$ is a pattern and $R$ is a relation (of any arity) over the points of $P$.
The augmented pattern $\tuple{P,R}$ will be denoted $P_R$.
\end{definition}

Obvious examples of
relations that could be added to a pattern are disequality relations or partial orders on
points, and this idea has been explored in a number of
papers~\cite{Cohen12:pivot,Cooper10:BTP,cz18:lmcs}. 

\begin{definition}
A homomorphism between augmented patterns $P_R$ and $P'_{R'}$ is a 
homomorphism $h$ from $P$ to $P'$ such that 
for all tuples $\tuple{x_1,x_2,\ldots,x_k} \in R$,
the tuple $\tuple{h(x_1),h(x_2),\ldots,h(x_k)} \in R'$.
\end{definition}
Using this extended definition of homomorphism, we can extend the notion of 
occurring as a sub-pattern (Definition~\ref{def:sub-pattern})
and occurring as a topological minor (Definition~\ref{def:top-minor}) to augmented patterns
in the natural way.

Now we can extend Definitions~\ref{def:CSPSP} and~\ref{def:CSPTM}, as follows,
to define restricted classes of CSP instances
and associated relations by forbidding the occurrence of certain augmented patterns.
\begin{definition}
\label{def:ForbSPaugmented}
Let $m$ be a constant, and let $\cal S$ be a set of augmented patterns 
such that for each $P_R \in {\cal S}$ the relation $R$ has arity $m$.
Let $\operatorname{Rel}$ be a partial function that maps an instance $I$
to a relation $R_I$ of arity $m$ over the points of $\PI{I}$.

We denote by \ForbSP{{\cal S},\operatorname{Rel}} the set
of all binary CSP instances $I$ such that $\operatorname{Rel}(I)$ is defined
and for all $P_R \in S$ it is not the case that $\SP{P_R}{\PI{I}_{\operatorname{Rel}(I)}}$.

We denote by \ForbTM{{\cal S},\operatorname{Rel}} the set
of all binary CSP instances $I$ such that $\operatorname{Rel}(I)$ is defined
and for all $P_R \in S$ it is not the case that $\TM{P_R}{\PI{I}_{\operatorname{Rel}(I)}}$.

\end{definition}

One of the simplest ways to augment a pattern $P$ is by adding a binary 
disequality relation, $\neq$, to specify that some points of $P$ are distinct.
A homomorphism from an augmented pattern $P_{\neq}$
to an augmented pattern $Q_{\neq}$ must map points that are specified to be distinct
in $P$ to points that are specified to be distinct in $Q$.
In the next three theorems, 
we shall assume that for any instance $I$, \emph{all} points in $\PI{I}_{\neq}$
are specified to be distinct.
In other words, we shall assume that for any instance $I$ 
the function $\operatorname{Rel}$ introduced
in Definition~\ref{def:ForbSPaugmented} always returns the binary relation $\neq$ 
containing all pairs of distinct points of $I$. We will denote this function by
$\operatorname{Rel}_{\neq}$. 

\begin{figure}[ht]
\thicklines \setlength{\unitlength}{0.7pt}
\centering
\begin{picture}(380,125) 
\put(0,80){\usebox{\varii}}
\put(80,80){\usebox{\varii}} \put(130,80){\usebox{\varii}}
\put(180,60){\usebox{\varii}} \put(230,80){\usebox{\varii}}
\put(280,80){\usebox{\varii}} \put(360,80){\usebox{\varii}}
\put(0,20){\usebox{\varii}} \put(80,20){\usebox{\varii}}
\put(130,20){\usebox{\varii}} 
\dashline[50]{7}(10,110)(35,100)
\dashline[50]{7}(65,100)(90,90) \dashline[50]{7}(90,110)(140,90)
\dashline[50]{7}(140,110)(190,90) \dashline[50]{7}(190,90)(240,110)
\dashline[50]{7}(240,90)(290,110) \dashline[50]{7}(290,90)(315,100)
\dashline[50]{7}(345,370,110) \dashline[50]{7}(10,30)(35,40)
\dashline[50]{7}(65,40)(90,50) \dashline[50]{7}(90,30)(140,50)
\dashline[50]{7}(140,30)(190,70) \dashline[50]{7}(345,100)(370,110)
\put(50,100){\makebox(0,0){.\;.\;.}} \put(50,40){\makebox(0,0){.\;.\;.}}
\put(330,100){\makebox(0,0){.\;.\;.}}

 \put(172,90){\makebox(0,0){$p$}} \put(172,70){\makebox(0,0){$q$}}
\put(258,47){\makebox(0,0){$p \neq q$}}
\end{picture}

\caption{The augmented pattern Pivot$_{\neq}$($k$).} 
\label{fig:pivotneq}
\setlength{\unitlength}{1pt}
\end{figure}
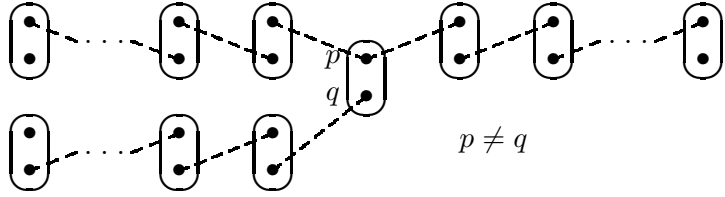

Now consider the augmented pattern Pivot$_{\neq}$($k$) which is obtained from 
the pattern Pivot($k$) defined in Definition~\ref{def:pivotk}
by adding a disequality relation specifying that the two points in the central node
are distinct, as shown in Figure~\ref{fig:pivotneq}.
Forbidding this pattern from occurring as a sub-pattern 
results in a larger class of instances than forbidding the pattern Pivot($k$),
but our next result shows that this larger class is still tractable.

\begin{theorem}
The augmented pattern Pivot$_{\neq}$($k$), shown in Figure~\ref{fig:pivotneq}, is sub-pattern tractable.
\end{theorem}

\begin{proof}
Let $I \in \ForbSP{Pivot_{\neq}(k),\operatorname{Rel}_{\neq}}$ for some constant $k$.
If $\PI{I}$ has a point $x_{v,a}$ which belongs to no negative
edge (i.e., it is compatible with all assignments to all other variables), 
then we can clearly remove all points in the same part as $x_{v,a}$
without introducing the pattern or affecting the existence of a solution. 
Thus we can assume without loss of generality that $\PI{I}$ contains no such points. 
A similar remark holds if $\PI{I}$ has any parts containing just a single point.

We can also assume without loss of generality that the constraint
graph of $I$ is connected. 
A variable $v$ is called an \emph{articulation variable} of $I$ if removing 
$v$ from $I$ disconnects the constraint graph of $I$. 
Any instance can be decomposed into a tree of components which only intersect at
articulation variables. 
It therefore suffices to show that any instance $I$ without
articulation variables can be solved in polynomial time,
so we shall assume that $I$ has no articulation variables. 

If Pivot($2k$) does not occur as a sub-pattern in $\PI{I}$ then, 
by Theorem~\ref{thm:pivotSPtractable} we have that $I$ is tractable.

To deal with the remaining case, 
assume that Pivot($2k$) occurs as a sub-pattern in $\PI{I}$ 
with the central part $U$ of Pivot($2k$) mapping to part $V$ of $\PI{I}$.
Let $S_{2k}$ be the set of parts of $\PI{I}$ to which the parts of Pivot($2k$) are mapped. 

Since Pivot$_{\neq}$($k$) does not occur as a sub-pattern in $\PI{I}_{\neq}$ 
(and hence neither does Pivot$_{\neq}$($2k$)), the two points in the central
part $U$ of Pivot($2k$) must map to the \emph{same} point in \PI{I}, 
which we denote by $x_{v,a}$. 

By our assumptions, we know that there is
another (distinct) value $b$ in the domain of $v$ which belongs to a negative
edge in \PI{I}, connecting part $V$ to some other part $W$.
If $W$ is only connected to $S_{2k}$ in the constraint graph of $\PI{I}$ 
via $V$, then $v$ is an articulation variable of $I$, which 
contradicts our assumption.
Hence, there is a path $\pi$ in the constraint graph of $\PI{I}$ 
linking $W$ to some part $Y \in S_{2k}$ such that $Y \neq V$.
 
By choosing $\pi$ to be minimal, we can assume that no other parts on the path $\pi$
belong to $S_{2k}$. Now, since $Y$ must lie on one of the three
branches of the occurrence of Pivot($2k$) in $\PI{I}$, we can extend $\pi$ by
following this branch from $Y$ either towards or away from the
central part $V$, in order to obtain a path of length at least $k$. 
This length-$k$ path, together with the first $k$ variables of 
the other two branches of Pivot($2k$), gives an occurrence of the pattern 
Pivot$_{\neq}$($k$) in $\PI{I}_{\neq}$, 
which contradicts our choice of $I$, so we are done.
\end{proof}

\begin{figure}[ht]
\thicklines \setlength{\unitlength}{0.7pt}
\centering

\begin{picture}(260,140)(20,-20)

\put(0,0){
\begin{picture}(100,110)(0,-10)
\put(0,60){\usebox{\vartwo}} \put(40,0){\usebox{\vartwo}}
\put(80,70){\usebox{\varone}} \dashline[50]{7}(10,70)(50,30)
\dashline[50]{7}(10,90)(90,90) \dashline[50]{7}(50,10)(90,90)
\put(55,-20){\makebox(0,0){$K_{\neq}$}}

\put(-8,90){\makebox(0,0){$p$}} \put(-8,70){\makebox(0,0){$q$}}
\put(32,30){\makebox(0,0){$p'$}} \put(32,10){\makebox(0,0){$q'$}}
\put(105,30){\makebox(0,0){\shortstack{ $p \neq q$ \\ $p' \neq q'$ }}}
\end{picture}
}

\put(180,0){
\begin{picture}(100,110)(0,-10)
\put(0,60){\usebox{\vartwo}} \put(40,0){\usebox{\vartwo}}
\put(80,60){\usebox{\vartwo}} \dashline[50]{6}(10,70)(50,30)
\dashline[50]{6}(10,90)(90,90) \dashline[50]{6}(50,10)(90,70)
\put(55,-20){\makebox(0,0){$\PG{C_3}_{\neq}$}}

\put(-8,90){\makebox(0,0){$p$}} \put(-8,70){\makebox(0,0){$q$}}
\put(32,30){\makebox(0,0){$p'$}} \put(32,10){\makebox(0,0){$q'$}}
\put(110,90){\makebox(0,0){$p''$}} \put(110,70){\makebox(0,0){$q''$}}
\put(130,20){\makebox(0,0){\shortstack{ $p \neq q$ \\ $p' \neq q'$ \\ $p'' \neq q''$ }}}
\end{picture}
}

\end{picture}

\caption{Two augmented patterns which are topological-minor tractable.}
\label{fig:neq}
\end{figure}
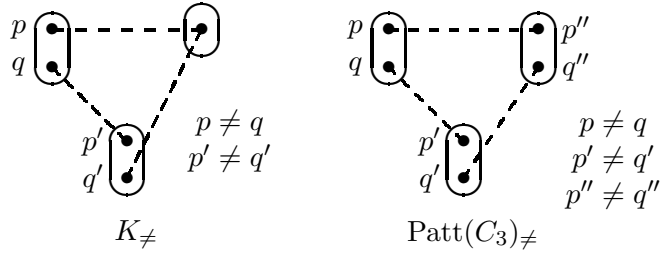

Now consider the augmented pattern $K_{\neq}$, shown in Figure~\ref{fig:neq}, 
which is obtained from the pattern $K$ shown in Figure~\ref{fig:new} 
by adding a disequality relation to specify that any two points in
the same part are distinct.
We now show that forbidding $K_{\neq}$ from occurring as a topological
minor results in a tractable class (which is larger than the class obtained by 
forbidding the pattern $K$ as a topological minor discussed in Theorem~\ref{thm:KisTMtractable}).

\begin{theorem}
\label{thm:KneqTMtractable}
The augmented pattern $K_{\neq}$, shown in Figure~\ref{fig:neq}, is sub-pattern NP-complete but topological-minor tractable.
\end{theorem}
\begin{proof}
By Theorem~\ref{thm:pivotSPtractable}, the (negative) pattern $K$ shown in
Figure~\ref{fig:new} is sub-pattern $\NP$-complete.
Since $\ForbSP{K}\subseteq\ForbSP{K_{\neq},\operatorname{Rel}_{\neq}}$, we have that $K_{\neq}$ is also 
sub-pattern $\NP$-complete. 

To show that $K_{\neq}$is topological-minor tractable
we will show that establishing arc-consistency is 
sufficient to decide the existence of a solution for any instance in $\ForbTM{K_{\neq},\operatorname{Rel}_{\neq}}$. 

By Lemma~\ref{lem:AC}, 
without loss of generality we need consider only arc consistent instances. 
We will show, by induction on the number of variables, that in any arc-consistent instance 
$I \in \ForbTM{K_{\neq},\operatorname{Rel}_{\neq}}$, any assignment to a single variable 
can be extended to a solution of $I$. 
This is certainly true for instances on up to two variables, by the
definition of arc consistency.

Now assume that $I$ has more than two variables,
and consider the assignment of the value $a$ to the variable $v$.
Let $I[v=a]$ be the
instance obtained from $I$ by making this assignment,
eliminating variable $v$ and eliminating from the domain of 
all other variables $w$ all values $b$ such that $\tuple{a,b} \notin R_{vw}$.
By arc consistency, none of the resulting domains in $I[v=a]$ is empty, 
i.e., for each variable $w$ there is a value $c_w$ in the domain of $w$
such that $\tuple{a,c_w} \in R_{vw}$. 
By the absence of $K_{\neq}$ as a topological minor in $\PI{I}_{\neq}$, 
we can deduce that all variables $w$ that were connected to $v$ in the constraint graph of $I$
are not connected in the constraint graph of $I[v=a]$. 

Let $S_1,\ldots,S_m$ be the connected components
of the constraint graph of $I[v=a]$. 
For any $k=1,\ldots,m$, consider the subinstance $I[S_k]$ 
of the original instance $I$ on the variables of $S_k$. 
Clearly, each $I[S_k] \in \ForbTM{K_{\neq},\operatorname{Rel}_{\neq}}$ and each $I[S_k]$ is arc-consistent.
Furthermore, since at least the variable $v$ has been eliminated
from the original set of variables, we know that each $I[S_k]$ has strictly
fewer variables than $I$ (even if $m=1$). 
Hence, by our inductive hypothesis,
the assignment of any value $c_w$ to any variable $w$ in $I[S_k]$
can be extended to a solution $s_k$ to $I[S_k]$. 
The solutions $s_k$ ($k=1,\ldots,m$) together with the assignment of $a$ to $v$ 
then form a solution to $I$ and the result follows by induction.
\end{proof}

Now consider the augmented pattern $\PG{C_3}_{\neq}$, shown in Figure~\ref{fig:neq}, 
which is obtained from the pattern $\PG{C_3}$ shown in Figure~\ref{fig:acyclic}
by adding a disequality relation specifying that any two points in the same part are distinct.
We now show that forbidding $\PG{C_3}_{\neq}$ from occurring as a topological
minor results in a tractable class (which is larger than the 
class of acyclic instances obtained by forbidding the pattern $\PG{C_3}$ 
as a topological minor discussed in Proposition~\ref{prop:acyclic}).

\begin{theorem}\label{thm:CneqTMtractable}
The augmented pattern $\PG{C_3}_{\neq}$, shown in Figure~\ref{fig:neq}, is sub-pattern NP-complete 
but topological-minor tractable.
\end{theorem}
\begin{proof}
By Theorem~\ref{thm:pivotSPtractable}, the (negative) pattern $\PG{C_3}$ shown in
Figure~\ref{fig:acyclic} is sub-pattern $\NP$-complete.
Since $\ForbSP{\PG{C_3}}\subseteq\ForbSP{\PG{C_3}_{\neq},\operatorname{Rel}_{\neq}}$, we have that $\PG{C_3}_{\neq}$ is
also sub-pattern $\NP$-complete. 

Singleton arc consistency (SAC) is an operation which consists in
applying the following operation on an instance $I$ until
convergence: if the instance $I[v=a]$ obtained by making the
assignment of the value $a$ to the variable $v$ and establishing arc consistency
is empty, then eliminate $a$ from the domain of $v$ in $I$.
To show that $\PG{C_3}_{\neq}$ is topological-minor tractable
we will show that SAC is a decision procedure for $\ForbTM{\PG{C_3}_{\neq},\operatorname{Rel}_{\neq}}$. 

Since establishing SAC cannot introduce any occurrence of the pattern, 
we need only consider instances that are singleton-arc-consistent 
(i.e., where no more eliminations are possible by SAC). 
We will show, by induction on the number of variables, that in any 
singleton-arc-consistent instance $I \in \ForbTM{\PG{C_3}_{\neq},\operatorname{Rel}_{\neq}}$,
any assignment to a single variable can be extended to a solution to $I$.
This is certainly true for instances on up to two variables, by the
definition of arc consistency.

Now assume that $I$ has more than two variables,
and consider the assignment of the value $a$ to the variable $v$.
Let $N$ be the set of parts of $\PI{I}$ that are connected by a negative edge to
$x_{v,a}$. We can assume that $N \neq \emptyset$, otherwise we could make
the assignment $a$ to variable $v$ without affecting the rest of the
instance $I$, and thus reduce $I$ to an instance on fewer variables
(which by our inductive hypothesis would have a solution).

Now let $I[N]$ be the subinstance of $I$ on the variables corresponding to 
parts in $N$, with the domain of each variable $w$ of $I[N]$ 
reduced to those values $c$ such that $(a,c) \in R_{vw}$. 
Since $I$ is singleton arc-consistent, $I[N]$ is arc-consistent. 

Let $J'_{\neq}$ be the augmented pattern shown in Figure~\ref{fig:proofP}.
\begin{figure}[ht]
\thicklines \setlength{\unitlength}{0.7pt}
\centering

\begin{picture}(460,110)(0,0)

\put(180,0){
\begin{picture}(100,110)(0,-10)
\put(0,60){\usebox{\vartwo}} \put(40,10){\usebox{\varone}}
\put(80,70){\usebox{\varone}} \dashline[50]{7}(10,70)(50,30)
\dashline[50]{7}(10,90)(90,90)
\put(55,-10){\makebox(0,0){$J'_{\neq}$}}

\put(-8,90){\makebox(0,0){$p$}} \put(-8,70){\makebox(0,0){$q$}} 
\put(70,30){\makebox(0,0){$r_1$}}
\put(140,50){\makebox(0,0){$p \neq q$}} 
\put(110,90){\makebox(0,0){$r_2$}}
\end{picture}
}

\end{picture}
\caption{The augmented pattern $J'_{\neq}$ used in the proof of Theorem~\ref{thm:CneqTMtractable}}
\label{fig:proofP}
\end{figure}
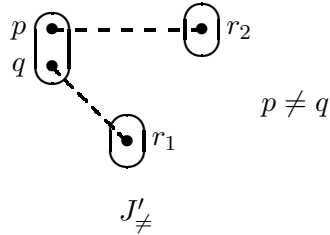
Note that $\SP{J'_{\neq}}{\PG{C_3}_{\neq}}$.
Now, since $\PG{C_3}_{\neq}$ does not occur as a topological minor in $\PI{I}_{\neq}$, 
we can deduce that $J'_{\neq}$ does not occur as a topological minor in $\PI{I[N]}$. 
Hence, $K_{\neq}$ does not occur as a topological minor in $\PI{I[N]}$ either, 
since $\SP{J'_{\neq}}{K_{\neq}}$. 
By the proof of Theorem~\ref{thm:KneqTMtractable}, 
any arc-consistent instance in $\ForbTM{K_{\neq},\operatorname{Rel}_{\neq}}$ has a solution, 
so $I[N]$ has a solution which we denote by $s_N$. 

Let $u$ be a variable of $I[N]$ and denote by $a_u$ the value assigned to $u$ by $s_N$. 
Let $I_u$ be the subinstance of $I$ on 
all variables of $I$ except $\{v\} \cup (N \setminus \{u\})$.

Let $S_u$ be the set of variables $w$ of $I_u$ which are either 
(1) $u$ itself, 
(2) directly constrained by the assignment of $a_u$ to $u$ 
(i.e., variables $w$ such that $\tuple{a_u,b} \notin R_{uw}$ for some $b$ in the domain of $w$),
or
(3) such that the pattern $J'_{\neq}$ occurs as a topological minor in $\PI{I_u}_{\neq}$ 
with the point $r_1$ of $J'_{\neq}$ mapping to $x_{u,a_u}$
and the point $r_2$ of $J'_{\neq}$ mapping to some point $x_{w,b}$ for some $b$.

Let $I[S_u]$ be the subinstance of $I$ on the set of variables $S_u$. 
Clearly $I[S_u]$ is singleton arc-consistent (since $I$ is), 
and has fewer variables than $I$ (since $v \notin S_u$). 
Hence, by our inductive hypothesis, the assignment
of value $a_u$ to variable $u$ can be extended to a solution $s_u$ of $I[S_u]$.

Now let $u' \in N \setminus \{u\}$. 
By the absence of $\PG{C_3}_{\neq}$ as a topological minor in $\PI{I}$, we can deduce
that no assignment in $s_u$ can be incompatible with any assignment
to a variable $y$ in $S_{u'} \setminus S_u$, except
possibly in the case that the assignment to $y$ is directly incompatible
with both the assignment of $a_u$ to $u$ and $a_{u'}$ to $u'$.
In this latter case, the solution $s_{u'}$ projected onto $S_{u'} \setminus S_u$ 
is necessarily consistent with $s_u$. 

Hence, by a simple inductive argument, we can create a consistent partial assignment
composed of the assignment of $a$ to $v$,
and the assignments specified by $s_N$ and each $s_u$ 
(projected onto the not-yet-assigned variables). 

The rest of the instance $I$, if it is non-empty, 
is not constrained by this partial assignment and by
our inductive hypothesis has a solution; combining these partial solutions 
gives a solution to $I$. 
\end{proof}

Classes of the CSP that are defined by specifying a restricted set of 
constraint relations over some fixed domain $D$ are known as \emph{language
classes}~\cite{DBLP:journals/jacm/JeavonsCG97,Feder98:monotone}.
Every known tractable language
class~\cite{DBLP:journals/jacm/JeavonsCG97,Barto14:survey} of CSP instances 
is characterised by an operation $f:D^k \rightarrow D$ with
the property that for all constraints $R_{uv}$, and all pairs
$\tuple{p_1,q_1},\tuple{p_2,q_2},\ldots,\tuple{p_k,q_k} \in R_{uv}$, the pair
$\tuple{f(p_1,p_2,\ldots,p_k),f(q_1,q_2,\ldots,q_k)} \in R_{uv}$; such an
operation is known as a \emph{polymorphism} of the constraint
relations~\cite{Barto14:survey,DBLP:journals/jacm/JeavonsCG97}.

We now show that using augmented patterns we can characterise every known
tractable language class using a single forbidden augmented sub-pattern.

\begin{theorem}
\label{thm:languageasForbSP}
Every tractable language class of binary CSP instances that is characterised by
a polymorphism $f$ 
is equal to $\ForbSP{P_R,\operatorname{Rel}_f}$ for some augmented pattern $P_R$ and 
function $\operatorname{Rel}_f$.
\end{theorem}
\begin{proof}
The $k$-ary operation $f:D^k \rightarrow D$ can be specified by a $(k+1)$-ary relation $R_f$ over $D$
where $R_f = \{\tuple{a_1,\ldots,a_{k+1}} \mid a_{k+1} = f(a_1,\ldots,a_k)\}$.
Define $\operatorname{Rel}_f$ to be the function that maps any CSP instance $I$ over $D$
to the relation $R$ over the points of $\PI{I}$, where 
$R = \{\tuple{x_{v,a_1},\ldots,x_{v,a_{k+1}}} \mid \tuple{a_1,\ldots,a_{k+1}} \in R_f\}$.

The class of all instances $I$ over domain $D$ 
for which all constraint relations admit $f$ as a 
polymorphism, is precisely the class of instances defined by 
$\ForbSP{P_R,\operatorname{Rel}_f}$ where $P = \tuple{X,E^\sim,E^+,E^-}$ with
\begin{itemize}
\item $X = U\cup V$, 
where $U = \{p_1,p_2,\ldots,p_{k+1}\}$ and $V = \{q_1,q_2,\ldots,q_{k+1}\}$;
\item $E^\sim = (U \times U) \cup (V \times V)$;
\item $E^+ = \{(p_i,q_i) \mid p_i \in U, q_i \in V, i=1,2,\ldots,k\}$;
\item $E^- = \{(p_{k+1},q_{k+1})\}$;
\end{itemize}
and $R = \{(p_1,p_2,\ldots,p_{k+1}),(q_1,q_2,\ldots,q_{k+1})\}$,
as illustrated in Figure~\ref{fig:tableau}.
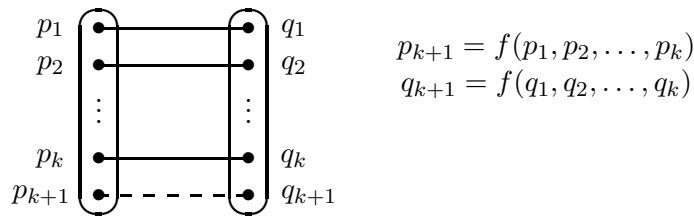
\begin{figure}[ht]
\centering
\thicklines \setlength{\unitlength}{0.7pt}
\newsavebox{\varkdotsbig}
\savebox{\varkdotsbig}(20,100){
\begin{picture}(20,100)(2,40)
\put(10,45){\oval(20,110)} 
\put(10,0){\makebox(0,0){$\bullet$}}
\put(10,20){\makebox(0,0){$\bullet$}} 
\put(10,50){\makebox(0,0){$\vdots$}}
\put(10,70){\makebox(0,0){$\bullet$}}
\put(10,90){\makebox(0,0){$\bullet$}}
\end{picture}}

\begin{picture}(150,120)(80,40)
\thicklines
\put(20,90){\usebox{\varkdotsbig}} 
\put(100,90){\usebox{\varkdotsbig}}
\put(30,140){\line(1,0){80}} 
\put(30,120){\line(1,0){80}} 
\put(30,70){\line(1,0){80}} 
\dashline[50]{7}(30,50)(110,50)

\put(5,140){\makebox(0,0){$p_1$}}
\put(5,120){\makebox(0,0){$p_2$}}
\put(5,70){\makebox(0,0){$p_k$}}
\put(0,50){\makebox(0,0){$p_{k+1}$}}
\put(135,140){\makebox(0,0){$q_1$}}
\put(135,120){\makebox(0,0){$q_2$}}
\put(135,70){\makebox(0,0){$q_k$}}
\put(143,50){\makebox(0,0){$q_{k+1}$}}

\put(270,130){\makebox(0,0){$p_{k+1} = f(p_1,p_2,\ldots,p_k)$}}
\put(270,110){\makebox(0,0){$q_{k+1} = f(q_1,q_2,\ldots,q_k)$}}
\end{picture}
\caption{The augmented pattern $P_R$ used in the proof of Theorem~\ref{thm:languageasForbSP}.}
\label{fig:tableau}
\end{figure}
\end{proof}

We remark that the algebraic dichotomy conjecture~\cite{Bulatov05:classifying},
which is a refinement of the dichotomy conjecture of Feder and
Vardi~\cite{Feder98:monotone}, implies that \emph{every} tractable language is
characterised by a single polymorphism, and thus under this conjecture
Theorem~\ref{thm:languageasForbSP} applies to \emph{all} tractable language
classes of binary CSP instance over a fixed domain.

\section{Conclusions and open problems}\label{sec:conclusion}

The notion of a pattern occurring as a topological minor, introduced here,
allows a new approach to the definition of tractable classes of CSP instances.
We have shown that this approach, together with the notion of augmented
patterns, can unify the description of all tractable structural and language
classes, as well as allowing new and more general tractable classes to be
identified. We therefore believe that it has great potential for systematically
identifying all tractable classes of the CSP. 

One long-term goal is to characterise precisely which patterns $P$
are topological-minor tractable and for which such patterns $P$,
\ForbTM{P} is recognisable in polynomial time. For
example, Figure~\ref{fig:open} shows three simple patterns whose
topological minor tractability is currently open. 

\thicklines \setlength{\unitlength}{0.7pt}
\begin{figure}[ht]
\centering

\begin{picture}(360,150)(0,0)

\put(0,50){
\begin{picture}(100,100)(0,0)
\put(0,60){\usebox{\vartwo}} \put(40,10){\usebox{\varone}}
\put(80,70){\usebox{\varone}} \dashline[50]{7}(10,70)(50,30)
\dashline[50]{7}(10,90)(90,90) \dashline[50]{7}(50,30)(90,90)
\end{picture}}

\put(20,0){
\begin{picture}(130,100)(0,0)
\put(0,0){\usebox{\varone}} \put(50,0){\usebox{\varone}}
\put(100,0){\usebox{\varone}} \put(150,0){\usebox{\varone}}
\put(200,0){\usebox{\varone}} \dashline[50]{7}(10,20)(60,20)
\dashline[50]{7}(60,20)(110,20) \dashline[50]{7}(110,20)(160,20)
\dashline[50]{7}(160,20)(210,20)
\end{picture}}

\put(220,20){
\begin{picture}(140,120)(0,0)
\put(0,80){\usebox{\varthree}} \put(60,0){\usebox{\varthree}}
\put(60,60){\usebox{\varthree}} \put(120,80){\usebox{\varthree}}
\dashline[50]{7}(10,90)(70,20) \dashline[50]{7}(70,10)(130,90)
\dashline[50]{7}(10,110)(130,110) \dashline[50]{7}(10,100)(70,90)
\dashline[50]{7}(70,80)(130,100) \dashline[50]{7}(70,30)(70,70)
\end{picture}}
\end{picture}

\caption{Three patterns whose topological-minor tractability is open.}

\label{fig:open}

\end{figure}
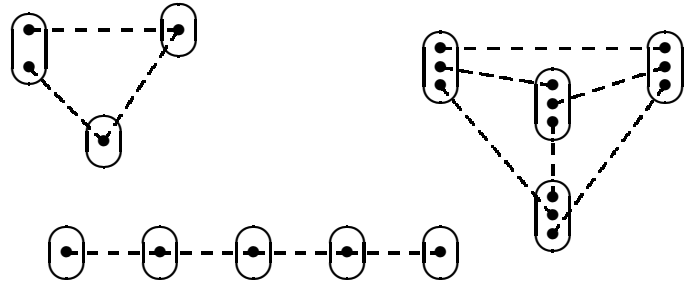

Another avenue of future
research is the discovery of other applications for topological
minors, such as in variable elimination~\cite{ccez15:jcss}.
Indeed, perhaps the most interesting open question is whether the notion of topological
minor, introduced in this paper, will find
applications other than the definition of tractable classes of the CSP.
We have seen that certain classic results from graph theory can lead
to results concerning topological minors of CSP instances. An
intriguing avenue for future research is to build bridges in the
other direction. For example, a corollary of the proof of
Theorem~\ref{thm:PXX} is that finding a path linking two given
vertices and which passes at most once through each part of an
$n$-partite graph is $\NP$-hard. Another way of expressing this is that
finding a \emph{heterochromatic path} linking two given vertices in a
vertex-coloured graph is $\NP$-hard~\cite{DBLP:journals/endm/LiZB01,broersma2005paths}.

To achieve further progress it may
well be necessary to further refine or modify the definition of a
topological minor given here. We regard this work as simply a first
step towards a general topological theory of complexity for
constraint satisfaction problems.

\newcommand{\noopsort}[1]{}\newcommand{\Zivny}{\noopsort{ZZ}\v{Z}ivn\'y}

\end{document}